\documentclass[letterpaper, USenglish, cleveref, autoref, thm-restate,numberwithinsect, nolineno]{socg-lipics-v2021}
\pdfoutput=1 
\hideLIPIcs  
\nolinenumbers 

\usepackage{amsmath}
\usepackage{amsthm}
\usepackage{amssymb}
\usepackage{mathtools}

\usepackage{graphicx}
\graphicspath{{./figures/}}

\usepackage{subcaption}

\usepackage{tikz}
\usepackage{tikz-cd}
\usepackage{blkarray}
\usetikzlibrary{babel,arrows,chains,matrix,positioning,scopes,calc,decorations, decorations.pathreplacing}
\usetikzlibrary{calc,arrows}

\usepackage{aliascnt}

\usepackage{comment}
\includecomment{clean}
\includecomment{precise}
\specialcomment{precise}{\color{gray}}{\color{black}}

\usepackage[colorinlistoftodos, prependcaption,textsize=tiny]{todonotes}
\setuptodonotes{color=blue!20!white, bordercolor=blue!20!white, textcolor=blue, noinline, noshadow}

\newcommand{\R}{\mathbb{R}}

\newcommand{\Set}{\mathbf{Set}}
\newcommand{\Merge}{\mathbf{Merge}}

\newcommand{\Forest}{\mathbf{Forest}}

\newcommand{\Top}{\mathbf{Top}}
\newcommand{\Vect}{\mathbf{Vect}}
\newcommand{\PP}{I}
\newcommand{\sd}{\hat{d}_I}
\newcommand{\WW}{\mathscr{W}}
\newcommand{\norm}[1]{\left\lVert#1\right\rVert}

\DeclareMathOperator{\B}{\mathcal{B}}
\DeclareMathOperator{\Mid}{mid}
\DeclareMathOperator{\Cells}{Cells}
\DeclareMathOperator{\id}{id}
\DeclareMathOperator{\LCA}{LCA}

\DeclareMathOperator{\costtext}{cost}
\newcommand{\cost}[1]{\costtext(#1)}

\newcommand{\define}[1]{\textbf{\boldmath{#1}}}

\colorlet{circle area}{pink!80}

\tikzset{
    filled/.style={fill=circle area, draw=circle edge, thick},
    outline/.style={draw=circle edge, thick},
    pinkcircle/.style={pink!80,fill opacity=0.5,draw=red!50 },
    emptycircle/.style={white,fill opacity=0.5,draw=red!50 }
}

\everymath{\displaystyle}

\title{The Universal $\ell^p$-Metric on Merge Trees} 

\author{Robert Cardona}{University at Albany, State University of New York (SUNY) }{rlcardona@albany.edu}{}{}

\author{Justin Curry\footnote{Corresponding Author}}{University at Albany, State University of New York (SUNY) }{jmcurry@albany.edu}{https://orcid.org/
0000-0003-2504-8388}{Supported by NSF CCF-1850052 and NASA 80GRC020C0016}

\author{Tung Lam}{University at Albany, State University of New York (SUNY) }{tlam@albany.edu}{}{}

\author{Michael Lesnick}{University at Albany, State University of New York (SUNY) }{mlesnick@albany.edu}{0000-0003-1924-3283}{}

\authorrunning{R. Cardona, J. Curry, T. Lam, and M. Lesnick} 

\Copyright{Robert Cardona, Justin Curry, Tung Lam, Michael Lesnick} 



\category{} 

\relatedversion{} 

\acknowledgements{
While H\aa vard Bjerkevik was not directly involved in this project, he has had a major influence on it, via his collaboration with ML on presentation distances for multiparameter persistence modules \cite{LB2021}.  
In particular, H\aa vard kindly agreed to share an early draft of \cite{LB2021} with our group in July 2020, which inspired many of the ideas in our paper.}

\EventEditors{}
\EventNoEds{2}
\EventLongTitle{38th International Symposium on Computational Geometry (SoCG 2022)}
\EventShortTitle{SoCG 2022}
\EventAcronym{SoCG}
\EventYear{2021}
\EventDate{June 7--10, 2022}
\EventLocation{Berlin, Germany}
\EventLogo{}
\SeriesVolume{}
\ArticleNo{}

\begin{document}

\maketitle

\begin{abstract}
Adapting a definition given by Bjerkevik and Lesnick \cite{LB2021} for multiparameter persistence modules, we introduce an $\ell^p$-type extension of the interleaving distance on merge trees.
We show that our distance is a metric, and that it upper-bounds the $p$-Wasserstein distance between the associated barcodes.
For each $p\in[1,\infty]$, we prove that this distance is stable with respect to cellular sublevel filtrations and that it is the universal (i.e., largest) distance satisfying this stability property.
In the $p=\infty$ case, this gives a novel proof of universality for the interleaving distance on merge trees.
\end{abstract}

\section{Introduction}

\subsection{Overview}

A \emph{merge tree}, also known as a barrier tree~\cite{flamm2002barrier} or a join tree \cite{carr2003computing}, encodes the connectivity of the sublevel sets of a function $f:X\to \R$ in terms of a graph $M_f$ equipped with a map $\pi:M_f \to \R$; see \Cref{fig:simple-example}.  
Merge trees are readily computed in practice, and have found applications in topography \cite{KWEON1994171, GC1987}, chemistry \cite{flamm2002barrier}, visualization \cite{carr2004simplifying,pont2021wasserstein}, cluster analysis \cite{hartigan81, JMLR-carlsson-memoli-10,chen2019generalized} and stochastic processes \cite{le2002random,perez2020c0persistent-a,perez2020persistent-b}.
As a fundamental topological descriptor, merge trees also play a central role in topological data analysis (TDA).

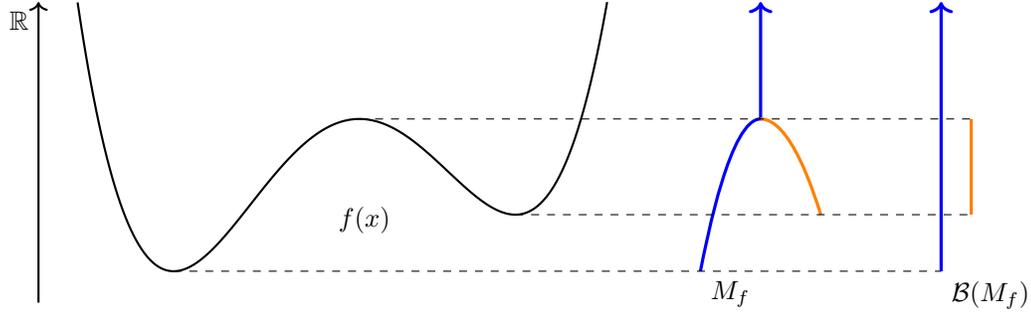
\begin{figure}
    \centering
\begin{tikzpicture}[scale=0.8]
    \draw [thick, domain=-2.35:6.45, samples=200] plot (\x, {pow(\x,4)/32 - 13*pow(\x, 3)/48 + 3*pow(\x,2)/8  + 13*\x/12});
    \node[anchor=south east] at (3, 0) {$f(x)$};
    \coordinate (A) at (0, -0.477);
    \coordinate (B) at (0, 0.463);
    \coordinate (C) at (0, 2.055);
    \coordinate (E) at (0, 4);
    \draw [blue, very thick, ->] ($(A) + (8, 0)$) node[black, anchor=north west] {$M_f$} parabola[bend at end] ($(C) + (9, 0)$) -- ($(E) + (9, 0)$);
    \draw [orange, very thick] ($(B) + (10, 0)$) parabola[bend at end] ($(C) + (9, 0)$);

    \draw [dashed] (-0.755, -0.477) -- ($(A) + (12, 0)$);
    \draw [dashed] (4.924, 0.463) -- ($(B) + (12.5, 0) $);
    \draw [dashed] (2.331, 2.055) -- ($(C) + (12.5, 0)$);

    \draw [blue, very thick, ->] ($(A) + (12,0)$) node[black, anchor=north west] {$\mathcal{B}(M_f)$} -- ($(E) + (12, 0)$) ;
    \draw [orange, very thick] ($(B) + (12.5,0)$) -- ($(C) + (12.5, 0)$);

    \draw [thick, ->] (-3, -1) -- (-3, 4) node[anchor = north east] {$\mathbb R$};
\end{tikzpicture}
\captionsetup{justification=centering}
\caption{\small Graph of function $f: \R \to \R$, its associated merge tree $M_f$ \\ and the {barcode $\mathcal B(M_f)$} obtained from $M_f$ via branch decomposition.}\label{fig:simple-example}
\end{figure}

Merge trees are closely related to \emph{persistent homology}, the most widely studied and applied TDA method. Persistent homology provides invariants of data called \emph{barcodes}; a barcode is simply a collection of intervals in $\R$. Each merge tree $M$ has an associated barcode $\B(M)$, which is obtained via a branch decomposition known as the \emph{elder rule}. This barcode is in fact the same as the \emph{sublevel set} persistence barcode in homological degree 0 considered in TDA~\cite{curry2018fiber}.

The question of how to metrize the collection of merge trees is a fundamental one: Metrics are needed to study the continuity and stability of the merge tree construction, and to quantify sensitivity to noise.
Many metrics on merge trees have been proposed in prior work, as we discuss in detail below.

In particular, metrics called \emph{interleaving distances}, which generalize the well-known \emph{Hausdorff distance} on subsets of a metric space, play a major role in TDA theory. Interleaving distances were first introduced by Chazal et al.~\cite{chazal2009proximity}; subsequently, the definition has been extended in several different directions \cite{bubenik2012:categorification,lesnick2015theory, bubenik2015metrics, scoccola2020:thesis, categorified_reeb, curry2014:thesis, kashiwara2018persistent, deSilva2017:catflow}.
Morozov et al.~observed that there is a natural definition of an interleaving distance for merge trees, denoted $d_I$, and used this to prove the following stability properties of merge trees and their barcodes \cite[Theorems 2 and 3]{interleaving-distance-merge-trees}:

\begin{theorem}[Stability properties of merge trees  \cite{interleaving-distance-merge-trees}]\label{thm:morozov-stability-review}\mbox{}
\begin{itemize}
    \item[(i)] For any functions $f,g:X \to \R$ where $X$ is a topological space, we have that
    \[d_I(M_f,N_g) \leq ||f-g||_{\infty}, \qquad \text{where} \qquad ||f-g||_{\infty}=\sup_{x\in X} |f(x)-g(x)|. \]
    \item[(ii)] For any merge trees $M$ and $N$, we have that \[d_B(\B(M),\B(N)) \leq d_I(M,N),\]
\end{itemize}
where $d_B$ denotes the bottleneck distance between barcodes; see \cref{def:wasserstein}.
\end{theorem}

While $d_B$ is the most common metric on barcodes in the TDA literature, it has a property that is undesirable in some settings: Informally, $d_B$ is sensitive only to the largest difference between two barcodes and not to smaller differences.  
To avoid such undesirable behavior,
 many applications of persistent homology and some theoretical works \cite{cohen2010lipschitz, turner2020medians, skraba2021wasserstein}
consider a generalization of $d_B$ called the \emph{$p$-Wasserstein distance}, denoted $d^p_\WW$; see \cpageref{def:wasserstein} for the definition.
Here $p\in [1,\infty]$ is a parameter and for $p=\infty$ we have that $d^\infty_{\WW}=d_B$.\footnote{There are various definitions of the $p$-Wasserstein distance in the TDA literature, which differ from each other by at most a factor of 2 \cite{cohen2010lipschitz, bubenik:algebraicwasserstein}.  In this paper, we use the version introduced by Robinson and Turner \cite{turner2017wasserstein}, and studied in \cite{skraba2021wasserstein,LB2021}.}
As the parameter $p$ decreases, the distances  $d^p_{\WW}$ become more sensitive to small differences between a pair of barcodes.
Because of this, the distances $d^p_{\WW}$ with small $p$, typically $p=1$ or $2$, are often preferred in practical applications; see \cite[Section 1]{LB2021} for a list of applications.

Both the bottleneck distance $d_B$ and the merge tree interleaving distance $d_I$ turn out to be instances of a general categorical definition of interleaving distances introduced by Bubenik and Scott \cite{bubenik2012:categorification}; this is shown for $d_B$ in \cite{BL2020:diagrams} and for $d_I$ in \cite[Proposition 3.11]{curry2021:DMT}.
As such, the two distances are closely related.  It is perhaps unsurprising, then, that the undesirable properties of $d_B$ mentioned above carry over to $d_I$: That is, $d_I$ is only sensitive to the largest difference between a pair of merge trees, and is insensitive to the smaller differences between them.

With this in mind, it is natural to ask whether we can define an $\ell^p$-type distance on merge trees analogous to the distance $d_\WW^p$ on barcodes, with similar theoretical properties as those given by \Cref{thm:morozov-stability-review}.
In this paper, we introduce such a distance, the \emph{$p$-presentation distance}, for each $p\in [1,\infty]$.
This distance is an analogue of the $\ell^p$-distance on multiparameter persistence modules recently introduced in \cite{LB2021}, and several of our main results are merge tree analogues of results from \cite{LB2021}.

To state an analogue of \Cref{thm:morozov-stability-review} for presentation distances, we will need some definitions:\label{def:cellular-monotone}
Let $X$ be a regular cell complex.
Following \cite{skraba2021wasserstein}, we say $f:X\to \R$ is \emph{cellular} if it is constant on each cell of $X$; and we say $f$ is \emph{monotone} if, in addition, its value on any cell $\sigma$ is greater than or equal to the values on $\partial \sigma$.  
Ordering the cells of $X$ arbitrarily, we may identify $f$ with an element of $\R^{|\Cells(X)|}$, so that the $\ell^p$-norm $\|f\|_p$ is well defined.

\begin{theorem}[$\ell^p$-stability properties of merge trees]\label{thm:main-theorem-summary}\mbox{}
For all $p\in [1,\infty]$,
\begin{itemize}
    \item[(i)] any pair of monotone cellular functions $f, g: X \to \R$ satisfies
    \[d_{\PP}^p(M_f,N_g) \leq ||f-g||_{p},\]
    \item[(ii)] any pair of merge trees $M$ and $N$ satisfies \[d_{\WW}^p(\B(M),\B(N)) \leq d_{\PP}^p(M, N).\]
\end{itemize}
\end{theorem}
\Cref{thm:main-theorem-summary} refines the degree-0 case of a fundamental $\ell^p$-stability result for persistent homology, due to Skraba and Turner \cite{skraba2021wasserstein}. 
We also establish the following universality result for $d_{\PP}^p$, which parallels a result on 1- and 2-parameter persistence modules proved in \cite{LB2021}.

\begin{theorem}\label{Thm:Universality}
For any $p\in [0,\infty]$, if $d$ is any metric on merge trees satisfying the stability property of \Cref{thm:main-theorem-summary}\,(i), then $d \leq d_{\PP}^p$.
\end{theorem}

Several $\ell^\infty$-type distances in the TDA literature have been shown to satisfy similar universal properties \cite{lesnick2015theory,blumberg2017universality,bauer2020reeb,bauer2020universality}.
In particular, \cite{bauer2020reeb} gives a universality result for a metric on Reeb graphs, which are closely related to merge trees.

In addition, we show the following:
\begin{theorem}\label{Thm:Infty_Pres_Dist_Equals_Interleaving_Distance}
$d_{\PP}^\infty=d_{I}$, i.e., the $\infty$-presentation distance and interleaving distance on merge trees are equal.
\end{theorem}

Together, \Cref{Thm:Universality,Thm:Infty_Pres_Dist_Equals_Interleaving_Distance} give us a universality result for the interleaving distance on merge trees.  A version of this result also appears in \cite{bauer2021quasi}, and was previously announced in a 2019 workshop talk \cite{bauer2019space}.  Whereas our paper only considers merge trees with finitely many nodes, \cite{bauer2021quasi} establishes universality of the interleaving distance for locally finite merge trees.  

In view of the good theoretical properties of the distances $d_{\PP}^{p}$, the question of whether these distances  can be efficiently computed is interesting.  Indeed, if they could be computed, then they could likely be used in practical applications in much the same ways that the Wasserstein distance on barcodes is commonly used.  It is known that computing $d_I=d_{\PP}^{\infty}$ on merge trees is NP hard \cite{agarwal2018computing}, but is fixed-parameter tractable \cite{farahbakhsh2019fpt};  
we would like to know whether these results extend to $d_{\PP}^{p}$ for all $p<\infty$, but we leave this to future work.

\subsection{Other metrics on merge trees}\label{sec:related-work}

While our $p$-presentation distance on merge trees is novel, many metrics on merge trees have been considered.
Recall that Morozov's interleaving distance $d_I$ \cite{interleaving-distance-merge-trees}, discussed above, is one example.  We mention several others: Various forms of edit distances on Merge trees have been proposed \cite{edit-distance-MTs-2020,pont2021wasserstein,sridharamurthy2021comparative}.  Since merge trees can be viewed as metric spaces in their own right, the \emph{Gromov-Hausdorff distance} on metric spaces can be used to compare merge trees \cite{agarwal2018computing}.  A Fr\'{e}chet-like distance between rooted trees was introduced in \cite{farahbakhsh2021fr}, along with an algorithm to compute this distance.  This distance was then applied to merge trees. 

The \emph{$p$-cophenetic distance} \cite{cardona2013cophenetic} is a metric on labeled merge tress which is similar to $d_{\PP}^{p}$;  see \Cref{def:cophenetic}.  In the $p=\infty$ case, an extension to unlabeled merge trees \cite{Munch2019:cophenetic,intrinsic_interleaving} was shown to be equal to $d_I$.  
Consequently, by \Cref{Thm:Infty_Pres_Dist_Equals_Interleaving_Distance}, the $\infty$-cophenetic distance and the $\infty$-presentation distance are the same.
However, for $p<\infty$, $d_{\PP}^{p}$ is a lower bound for the $p$-cophenetic distance.
\Cref{ex:instability-cophenetic} illustrates how they differ and demonstrates that the $p$-cophenetic distance lacks the stability of \cref{thm:main-theorem-summary}\,(i).

In addition, several metrics on \emph{Reeb graphs} have been studied; since the geometric realization (\cref{defn:geometric-realization}) of every merge tree is a Reeb graph, any metric on Reeb graphs specializes to a metric on merge trees.  A definition of interleavings different from that in \cite{interleaving-distance-merge-trees} was used to define a metric on Reeb graphs in \cite{categorified_reeb}.  A family of truncated interleaving distances generalizing this was introduced in  \cite{chambers2021:truncated}. 
The \emph{functional distortion distance} on Reeb graphs \cite{bauer2014measuring} was shown to satisfy stability properties analogous to \cref{thm:morozov-stability-review} and to be strongly equivalent to the interleaving distance \cite{bauer2015strong}.
Edit distances on Reeb graphs were defined in \cite{di2016edit,bauer2020reeb},  
and \cite{bauer2020reeb} showed that its Reeb graph edit distance is universal.  
Recent work \cite{bollen2021reeb} surveys these metrics on Reeb graphs and their relationships.  
Finally, the \emph{contortion distance} \cite{bauer2021quasi} was shown to be strongly equivalent to each of the distances considered in \cite{categorified_reeb}, \cite{bauer2014measuring}, and \cite{bauer2020reeb}, and to be universal on contour trees.

\section{Merge trees}
In this section, we define merge trees.  We work primarily with a categorical definition, which is convenient for defining the interleaving and presentation distances.  
Recall that we may regard any partially ordered set $(P,\preceq)$ as a category with one object for each element $p\in P$ and a morphism from $p$ to $q$ whenever $p\preceq q$.
We will be particularly interested in the  posets $\R$ and $[n]=\{0,1,\dots,n\}$.

\begin{definition}
Given a topological space $X$ and function $f\colon X\to \R$, the sublevel set filtration of $f$ is the functor $S^{\uparrow}f\colon \R \to \Top$ given by $S^{\uparrow}f(t)=f^{-1}(-\infty,t]$, with $S^{\uparrow}f(s\leq t)$ the inclusion $S^{\uparrow}f(s)\hookrightarrow S^{\uparrow}f(t)$.
\end{definition}

\begin{definition}[cf.~\cite{patel2018generalized,curry2018fiber}]\label{def:merge-tree}
A \define{persistent set} is a functor $M\colon \R \to \Set$. 
\end{definition}

\begin{example}
Letting $\pi_0\colon \Top \to \Set$ denote the connected components functor, the composition $\pi_0\circ S^{\uparrow}f$ is a persistent set.
\end{example}

\begin{definition}[cf.~\cite{patel2018generalized}]
We say a persistent set $M$ is \define{constructible} if there exists a set
\[\tau := \{s_0 < s_1 < \cdots < s_{n}\}\subset \R \qquad \text{such that} \]  
\begin{enumerate}
\item If $M\ne \emptyset$, then $\tau\ne \emptyset$ and $M(t)=\emptyset$ for all $t<s_0$,
\item $M(s\leq t)$ is an isomorphism whenever $s,t\in [s_i,s_{i+1})$, and also for $s,t\in [s_n,\infty)$.
\end{enumerate}

If $M$ is constructible, then we call the minimal such $\tau$ the set of \define{critical times} of $M$, and denote it $\tau_M$.  For $s_i\in \tau$, we call $M(s_i)$ a \define{critical set}.
\end{definition}
Note that to specify a constructible merge tree $M$ (up to ismorphism), it suffices to specify $\tau_M$, the critical sets $M(s_i)$, and the functions $m_i:=M(s_i\leq s_{i+1})$.

\begin{remark} Equivalently, one can define constructibility using categorical language:  A persistent set $M$ is constructible if it is  isomorphic to the left Kan extension of some functor $M':[n] \to \Set$ along an order preserving map $j \colon [n] \hookrightarrow \mathbb R$.
\end{remark}

\begin{definition}
A \define{merge forest} is a constructible persistent set $M$ where each $M(t)$ is finite.
A \define{merge tree} is a merge forest where $M(t)=\{*\}$ for $t$ sufficiently large.
\end{definition}
We denote the category of merge forests by $\Forest$ and the category of merge trees by $\Merge$.
The categorical perspective on merge trees was previously considered in \cite{bubenik2015metrics,patel2018generalized,curry2018fiber} and offers the advantage of a streamlined definition of the interleaving distance; \Cref{def:interleaving}.

\begin{remark}
Applying the usual disjoint union of sets at each index, we obtain a well defined notion of disjoint union of persistent sets.  The disjoint union of finitely many merge forests is itself a merge forest.
\end{remark}

The next example is fundamental, as it provides a notion of generators for merge trees.

\begin{example}\label{ex:strand}
A \define{(closed) strand born at $s$}, written $F_s : \R \to \Set$, is the persistent set 
\[
F_s(t) :=
\begin{cases}
\emptyset & \text{if } t < s, \\
\{*\} & \text{if } t \geq s.
\end{cases}
\]
Closed strands are clearly constructible.  
The analogous \define{open strands}, where $F_s(t)=\{*\}$ if and only if $t>s$, are not constructible.
Our strands will always be closed strands.
For $m\in \mathbb{Z}_{>0}$, we let $F^m_s$ denote the disjoint union of $m$ copies of $F_s$.
\end{example}

Following \cite{curry2021:DMT}, we define the geometric realization of a constructible persistent set; this relates our categorical definition of a merge tree to the topological and graph-theoretic definitions that one typically sees in the literature.
Our definition is equivalent to that of \cite[Definition A.3]{curry2021:DMT} in the constructible setting, though slightly different in the details.

\begin{definition}[Geometric realization]\label{defn:geometric-realization}
  Given a constructible persistent set $M$, the \define{geometric realization} of $M$ is a pair $|M|=(X,\gamma)$ where $X$ is a topological space and $\gamma:X\to \R$ is a continuous function.
We take the set underlying $X$ to be $\sqcup_{t\in \R} M_t$, and for $x\in M_t$, we define $\gamma(x)=t$.  It remains to specify the topology on $X$.  We regard $X$ as a poset, with $x \preccurlyeq y$ iff both $\gamma(x)\leq \gamma(y)$ and $M(\gamma(x)\leq \gamma(y))(x)=y$.  For $x\in X$, let
\[C_x=\{y\in X\mid \textup{$y$ and $x$ are comparable}\}.\]
We declare a set $U\subset X$ to be open if and only if for each $x\in U$, there exists $V\subset \R$ open such that $(\gamma^{-1}(V)\cap C_x)\subset U$.
It is easily verified that this indeed defines a topology on $X$, and that $\gamma$ is continuous with respect to this topology.
\end{definition}

Abusing terminology slightly, we say that a function $f:Y\to \R$ is \define{constructible} if $\pi_0\circ S^{\uparrow}f$ is constructible.  One can check that if $f$ is constructible and continuous, 
then $|\pi_0 \circ S^{\uparrow}f|$ is equivalent to the topological construction of a merge tree considered in \cite{interleaving-distance-merge-trees}. 

A point $u$ is an \define{ancestor} of a point $v$ if  $v \preccurlyeq u$.  The \define{least common ancestor} $\LCA(v,w)$ of nodes $v$ and $u$ is the common ancestor of $v$ and $w$ with minimal height.  $\LCA(v,w)$ may not exist, but if it exists it is unique.

\begin{example}\label{example-m}
We specify a merge tree $M$ by taking $\tau_M=\{0,1,2,3,5\}$,
$M_0 := \{0\}$, $M_1 := \{0, 1\}$, $M_2 := \{0, 1, 2\}$, $M_3 = \{0, 2\}$, and $M_5 = \{0\}$,
\[
m_2 : \begin{cases}0 \mapsto 0 \\ 1 \mapsto 0 \\ 2 \mapsto 2 \end{cases}
\qquad \text{and} \qquad
m_3 : \begin{cases} 0 \mapsto 0 \\ 2 \mapsto 0 \end{cases}
\]
and the remaining $m_i$ to be inclusions.  The diagram
\[
\begin{tikzcd}
M_0 \ar[r, "m_0"] & M_1 \ar[r, "m_1"] & M_2 \ar[r, "m_2"] & M_3 \ar[rr, "m_3"] & & M_5
\end{tikzcd}
\]
can be represented pictorially by:

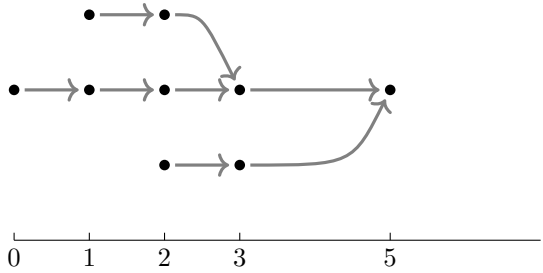
\[
    \begin{tikzpicture}[->, very thick, gray]
        \coordinate (02) at (0, 2);
        \coordinate (12) at (1, 2);
        \coordinate (13) at (1, 3);
        \coordinate (21) at (2, 1);
        \coordinate (22) at (2, 2);
        \coordinate (23) at (2, 3);
        \coordinate (31) at (3, 1);
        \coordinate (32) at (3, 2);
        \coordinate (52) at (5, 2);

        \draw[shorten >=4pt, shorten <=4pt] (02) -- (12);
        \draw[shorten >=4pt, shorten <=4pt] (12) -- (22);
        \draw[shorten >=4pt, shorten <=4pt] (22) -- (32);
        \draw[shorten >=4pt, shorten <=4pt](32) -- (52);
        \draw[shorten >=4pt, shorten <=4pt] (13) -- (23);
        \draw[shorten >=4pt, shorten <=4pt] (23) .. controls (2.5, 3) .. (32);
        \draw[shorten >=4pt, shorten <=4pt] (21) -- (31);
        \draw[shorten >=4pt, shorten <=4pt] (31) .. controls (4.5, 1) .. (52);

        \foreach \x in {0,1,2,3,5} {
        \draw[-, black, thin] (\x ,0) node[anchor=north] {\x} -- (\x,0.1);
    };
        \draw[black, -, thin] (0, 0) -- (7, 0);
        \foreach \x in {(02), (12), (13), (21), (22), (23), (31), (32), (52)}{
        \fill[black] \x circle[radius=2pt];
    };
    \end{tikzpicture}
\]
The geometric realization of $M$ can be drawn as follows:

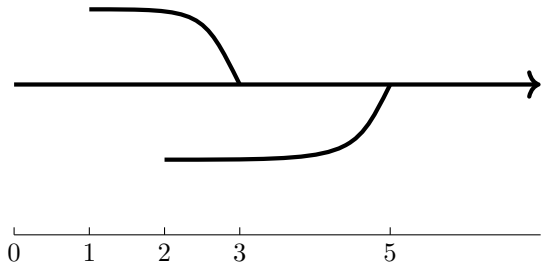
\[
    \begin{tikzpicture}
    \draw [ultra thick, ->] (0, 4) -- (7, 4);
    \draw [ultra thick] (1, 5) .. controls (2.5, 5) .. (3, 4);
    \draw [ultra thick] (2, 3) .. controls (4.5, 3) .. (5, 4);
    \foreach \x in {0,1,2,3,5} {
        \draw (\x ,2) node[anchor=north] {\x} -- (\x,2.1);
    };
    \draw (0, 2) -- (7, 2);
\end{tikzpicture}
\]
\end{example}

\paragraph*{Barcodes of merge trees}

Fix a field $k$.  A \define{persistence module} is a functor $N:\R\to \Vect$, where $\Vect$ denotes the category of $k$-vector spaces.
Every persistent set $M: \R \to \Set$ has an associated persistent homology module $H_0(M):\R \to \Vect$ given by the free functor from $\Set$ to $\Vect$.

Recall that a barcode is a multiset of intervals in $\R$.  According to a well-known structure theorem \cite{crawley2015decomposition}, there is a unique barcode $B(N)$ associated to any pointwise finite dimensional (PFD) persistence module $N$.
Thus, any merge tree $M$ has a well-defined barcode $B(M):=B(H_0(M))$.  
As suggested in the introduction, $B(M)=B(H_0(S^{\uparrow}f))$ where $|M| = (X, f)$ is the geometric realization of $M$.
Moreover, if $f:X\to \R$ is a constructible continuous function, then $B(\pi_0 \circ S^{\uparrow}f)=B(H_0(S^{\uparrow}f))$ \cite{curry2018fiber}.

\section{Presentations of merge trees}

We now define presentations of merge trees.
Recall that, in the usual algebraic setting, presentations are defined in terms of generators and relations.
For merge trees, generators correspond to closed strands, as defined in \Cref{ex:strand}, which can be thought of as starting branches that emanate from each leaf node. 
An internal/merge node above leaf nodes $i$ and $j$ witnesses the equality $\text{branch}_i=\text{branch}_j$ by mapping a relating strand into the generating strands for $i$ and $j$.  
Each presentation $P_M$ of a merge tree $M$ then gives rise to a matrix that encodes which generators are identified by which relation.
We show that any pair of merge trees can be given presentations so that these matrices are identical---we call these "compatible presentations."
\begin{definition}[Coequalizer]
Given sets $A, B$, the \define{coequalizer} of a pair of functions $\alpha, \beta \colon A \rightrightarrows B $ is the set of equivalence classes $C:=B/\sim$, where
$\sim$ is the equivalence relation on $B$ generated taking $\alpha(x)\sim\beta(x)$ for all $x\in A$.
Let $q:B\to C$ denote the quotient map.  $C$ satisfies the following universal property: Given a  function $\gamma: B \to D$ such that $\gamma\circ \beta = \gamma \circ \alpha$, there is a unique function $\delta: C \to D$ such that $\gamma=\delta\circ q$.
The definition of a coequalizer extends pointwise to persistent sets: Given persistent sets $A$ and $B$, the \define{coequalizer} of a pair of natural transformations $\alpha, \beta : A \rightrightarrows B $ is the persistent set $C$ such that for all $t\in \R$, $C(t)$ is the coequalizer of $\alpha(t), \beta(t) \colon A(t) \rightrightarrows B(t)$, with the maps internal to $C$ given by universal properties.
The universal property of coequalizers of persistent sets is completely analogous. 
\end{definition}
\begin{definition}[Presentation]\label{def:presentation}
A collection of strands $\{G_i\}$ and $\{R_j\}$---called \define{generators} and \define{relations}, respectively---and merge functions $f_j, g_j : R_j \to G := \sqcup_i G_i$ define a \define{presentation} of $M$ if $M$ is the coequalizer in $\Forest$ of the following diagram, written $P_M$,
\[
\begin{tikzcd}
\bigsqcup_j R_j  \ar[r, yshift=.7ex, "f"] \ar[r, yshift=-.7ex, "g"'] &
\bigsqcup_i G_i \ar[r, dashed] & M.
\end{tikzcd}
\]
The maps $f$ and $g$ are induced by the merge functions $f_j, g_j$ and the universal mapping property of disjoint unions.
\end{definition}

Although this definition is cloaked in category theory, coequalizers and disjoint unions formalize gluing constructions in topology, which are more generally cast in terms of colimits.
Intuitively, relating strands indicate where generating strands are glued together; see \Cref{fig:fork2}.
    \begin{figure}[h!]
        \centering
        \resizebox{.8\textwidth}{!}{
    \begin{tikzpicture}
        \draw [->, ultra thick] (0, 3) node[above] {$G_1$} --  (4, 3);
        \draw [->, ultra thick] (1, 1) node[below] {$G_2$} --  (4, 1);
        \draw [->, ultra thick, brown] (2, 2) -- node[below] {$R$} (4, 2);
        \draw [->, shorten >=2pt, shorten <= 2pt] (2, 2) to [bend left] node[midway, left] {$f$} (2, 3);
        \draw [->, shorten >=2pt, shorten <= 2pt] (2, 2 - 0.05) to [bend right] node[midway, left] {$g$} (2, 1);
        \draw [->, thick, brown] (2,3 - 0.1) -- (4, 3 - 0.1);
        \draw [->, thick, brown] (2,1 + 0.1) -- (4, 1 + 0.1);
        \draw [->, dashed, thick] (5, 1.5) -- (6, 1.5);
        \foreach \x in {0,1,2, 3, 4} {
        \draw (\x,0) node[anchor=north] {\x} -- (\x,0.1);
        \draw (\x + 7 ,0) node[anchor=north] {\x} -- (\x + 7,0.1);
         };
        \draw (0, 0) -- (4, 0);
        \draw (7, 0) -- (11, 0);
        \draw [ultra thick] (7, 3) .. controls (8.5, 3) .. (9, 2);
        \draw [ultra thick] (8, 1) .. controls (8.5, 1) .. (9, 2);
        \draw [ultra thick, ->] (9, 2) to (11, 2);
    \end{tikzpicture}}
        \caption{Presenting a merge tree with two branches as a coequalizer.}
        \label{fig:fork2}
    \end{figure}
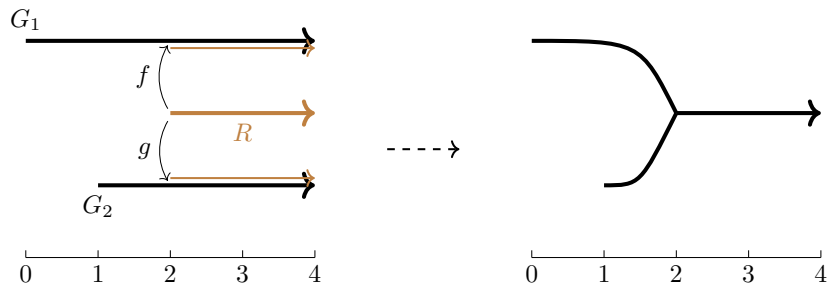

We can encode a presentation by the presentation matrix and its label vector.
\begin{definition}[Presentation matrix and label vector]
Given a presentation $P_M$ of a merge tree $M$ with $k$ generators and $l$ relations,
\[
\begin{tikzcd}
\bigsqcup_{j = 1}^l R_j  \ar[r, yshift=.7ex, "f"] \ar[r, yshift=-.7ex, "g"'] &
\bigsqcup_{i = 1}^k G_i \ar[r, dashed] & M,
\end{tikzcd}
\]
we can pick an ordering of generators and relations to obtain a $(k \times l)$ \define{presentation matrix} where the $i$-th row corresponds to the $i$-th generator $G_i$ and the $j$-th column corresponds to the $j$-th relation $R_j$.
The $(i, j)$-entry of the presentation matrix is 1 if either of $f_j, g_j : R_j \to G$ maps to $G_i$ and 0 otherwise.

The \define{label vector} of $P_M$ is the $(k+l)$-vector $L(P_M)$ where the first $k$ entries encode the birth times of each generating strand and the remaining $l$ entries encode the birth times of each of the relating strands. We will separate the row (generator) birth times from the column (relation) birth times by a semi-colon for legibility.
\end{definition}

We now give several examples of presentations, presentation matrices, and their label vectors.
We will see in particular that presentations are \emph{not} unique.
Indeed, one can always modify an existing presentation by introducing an extra generating strand that is then killed by an identical relating strand, as in \Cref{presentation-example,trivial-tree}. 
\Cref{ex:diff-merge-fns} shows that once two generators have been merged, any further merge event can be encoded by a merge function that maps to either generator.

\begin{example}[Tree with one leaf node]\label{trivial-tree}
If $M$ is a merge tree with one leaf node born at time $s$, then there is a presentation $P_M$ given by
\[
\begin{tikzcd}
F_s  \ar[r, yshift=.7ex, "\id"] \ar[r, yshift=-.7ex, "\id"'] &
F_s \ar[r, dashed] & M.
\end{tikzcd}
\]
The corresponding presentation matrix and label vector are, respectively,
        \[
        \begin{blockarray}{cc}
                & s  \\
            \begin{block}{c(c)}
            s & 1 \\
            \end{block}
        \end{blockarray}\qquad \text{ and }\qquad
        L(P_M) = [s; s].
        \]
        One can also obtain a presentation $P'_M$ of $M$ with one generator $F_s$ and no relations.
        In this case, the presentation matrix of $P'_M$ is an (empty) $1 \times 0$ matrix, but whose label vector is $L(P'_M)=[s]$.
\end{example}
\begin{example}[Adding a trivial generator and relation] \label{presentation-example}
Let $M$ be a merge tree with two leaves born at times $0$ and $1$ that merge at time $2$.
The presentation $P_M$ in \Cref{fig:fork2} uses two generators and one relation
\[
\begin{tikzcd}
F_2  \ar[r, yshift=.7ex, "f"] \ar[r, yshift=-.7ex, "g"'] &
F_0 \sqcup F_1\ar[r, dashed] & M,
\end{tikzcd}
\]
where $f: F_2 \to F_0$, and $g: F_2 \to F_1$ are Merge functions.
One can modify $P_M$ to obtain a presentation $P'_M$ by introducing an extra generator and relation $F_a$, for any $a \in [1, 2)$,
\[
\begin{tikzcd}
F_a\sqcup F_2  \ar[r, yshift=.7ex, "f_1\sqcup f_2"] \ar[r, yshift=-.7ex, "g_1\sqcup g_2"'] &
F_0 \sqcup F_1 \sqcup F_a \ar[r, dashed] & M,
\end{tikzcd}
\]
with $f_1: F_a \to F_a$, $g_1: F_a \to F_1$ and $f_2: F_2 \to F_1$, $g_2: F_2 \to F_0$; see \Cref{fig:add-trivial-relation}.
The corresponding presentation matrices and label vectors for $P_M$ and $P'_M$ are then, respectively,

\[
            \begin{blockarray}{cc}
                & 2  \\
            \begin{block}{c(c)}
            0 & 1 \\
            1 & 1 \\
            \end{block}
            \end{blockarray}
            \quad L(P_M)=[0,1;2]
\quad \text{ and }\quad
\begin{blockarray}{ccc}
                & 2 & a \\
            \begin{block}{c(cc)}
            0 & 1 & 0 \\
            1 & 1 & 1 \\
            a & 0 & 1 \\
            \end{block}
        \end{blockarray}
        \quad L(P'_M)=[0,1,a;2,a].
\]
\begin{figure}[h!]
    \centering
    \begin{tikzpicture}
        \draw [->, ultra thick] [shift={(0, 2)}] (0, 3) node[above] {$F_0$} --  (4, 3);
        \draw [->, ultra thick] [shift={(0, 2)}]  (1, 1) node[below] {$F_1$} --  (4, 1);
        \draw [->, ultra thick, brown] [shift={(0, 2)}]  (2, 2) -- node[below] {$F_2$} (4, 2);
        \draw [->, brown] [shift={(0, 2)}]  (2 - 0.05, 2 + 0.05) to [bend left] node[midway, left] {$f$} (2, 3 - 0.1);
        \draw [->, brown] [shift={(0, 2)}] (2- 0.05, 2 - 0.05) to [bend right] node[midway, left] {$g$} (2, 1 + 0.1);
        \foreach \x in {0,1,2, 3, 4} {
        \draw (\x,0) node[anchor=north] {\x} -- (\x,0.1);
        \draw [shift={(7, 0)}] (\x ,0) node[anchor=north] {\x} -- (\x,0.1);
         };
        \draw (0, 0) -- (4, 0);
        \draw [shift={(7, 0)}] (0, 0) -- (4, 0);

        \draw [->, ultra thick] [shift={(7, 2)}] (0, 3) node[above] {$F_0$} --  (4, 3);
        \draw [->, ultra thick] [shift={(7, 2)}]  (1, 1) node[above, left] {$F_1$} --  (4, 1);
        \draw [->, ultra thick, blue] [shift={(7, 2)}]  (2, 2) -- node[below] {$F_2$} (4, 2);
        \draw [->, blue] [shift={(7, 2)}]  (2 - 0.05, 2 + 0.05) to [bend left] node[midway, left] {$g_2$} (2, 3 - 0.1);
        \draw [->, blue] [shift={(7, 2)}] (2- 0.05, 2 - 0.05) to [bend right] node[midway, left] {$f_2$} (2, 1 + 0.1);
        \draw[->, ultra thick, orange] [shift={(7,0)}] (1.5, 2) -- node[below] {$F_a$} (4, 2) ;
        \draw[->, ultra thick] [shift={(7,0)}] (1.5, 1) node[below] {$F_a$} -- (4, 1);
        \draw[->, orange] [shift={(7,0)}] (1.5 - 0.05, 2) to [bend left] node[midway, left] {$g_1$} (1.5, 3-0.1);
        \draw[->, orange] [shift={(7,0)}] (1.5 - 0.05, 2) to [bend right] node[midway, left] {$f_1$} (1.5, 1 + 0.1 );
    \end{tikzpicture}
    \caption{Presentations $P_M$ (left) and $P'_M$ (right) of the merge tree from \Cref{fig:fork2}.}
    \label{fig:add-trivial-relation}
\end{figure}
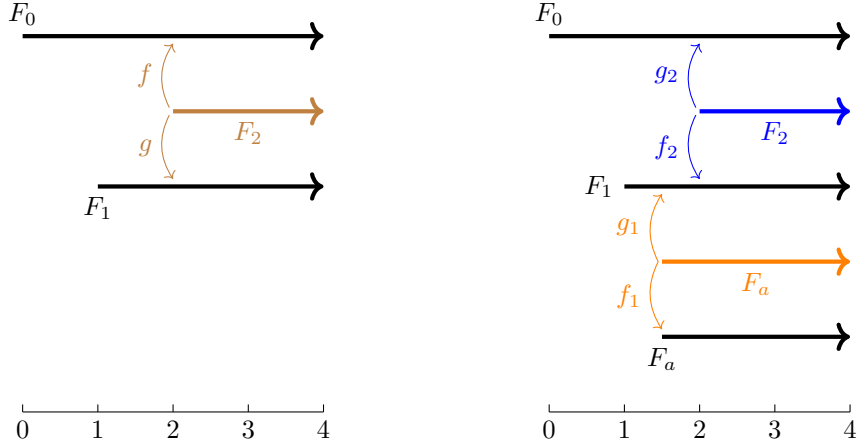
\end{example}

\begin{example}[Different merge functions]\label{ex:diff-merge-fns}
Consider the merge tree from \Cref{example-m} with three leaf nodes and two least common ancestors.
Since one least common ancestor occurs after the other, there is a choice as to which generator you choose to attach to.
\Cref{fig:3branch-tree-presentations} shows these two possible choices.
    \begin{figure}[h!]
        \centering
        \resizebox{\textwidth}{!}{%
            \begin{tikzpicture}
        \foreach \x in {0,1,2,3,5} {
        \draw (\x ,0) node[anchor=north] {\x} -- (\x,0.1);
        \draw [shift={(9, 0)}]  (\x ,0) node[anchor=north] {\x} -- (\x,0.1);
        };
        \draw (0, 0) -- (7, 0);
        \draw (9, 0) -- (16, 0);
        \draw[ultra thick, ->] (1, 5) node[left] {$G_1$} -- (7, 5);
        \draw[ultra thick, ->] (0, 3) node[left, below] {$G_2$} -- (7, 3);
        \draw[ultra thick, ->] (3, 1) node[left, below] {$G_3$} -- (7, 1);
        \draw[ultra thick, ->, orange] (3, 4) -- (7, 4) node[below] {$R_1$} ;
        \draw[ultra thick, ->, blue] (5, 2)  -- (7, 2) node[below] {$R_2$} ;
        \draw[->, orange] (3, 4) to [bend left] node[midway, left] {$f_1$} (3, 5);
        \draw[->, orange] (3, 4) to [bend right] node[midway, left] {$g_1$} (3, 3);
        \draw[ ->, blue] (5, 2)  to [bend left] node[midway, left] {$f_2$} (5, 3);
        \draw[ ->, blue] (5, 2)  to [bend right] node[midway, left] {$g_2$} (5, 1);
        \node at (1, 1) {(a)};
        \draw[ultra thick, ->] [shift={(9, 0)}] (1, 5) node[left] {$G_1$} -- (7, 5);
        \draw[ultra thick, ->] [shift={(9, 0)}] (0, 3) node[left, below] {$G_2$} -- (7, 3);
        \draw[ultra thick, ->] [shift={(9, 0)}] (3, 1) node[left, below] {$G_3$} -- (7, 1);
        \draw[ultra thick, ->, orange] [shift={(9, 0)}]  (3, 4) -- (7, 4) node[below] {$R_1$} ;
        \draw[ultra thick, ->, blue] [shift={(9, 0)}] (5, 2)  -- (7, 2) node[below] {$R_2$} ;
        \draw[->, orange] [shift={(9, 0)}]  (3, 4) to [bend left] node[midway, left] {$f_1$} (3, 5);
        \draw[->, orange] [shift={(9, 0)}]  (3, 4) to [bend right] node[midway, left] {$g_1$} (3, 3);
        \draw[ ->, blue] [shift={(9, 0)}]  (5, 2)  to [bend left] node[midway, left] {$f_2$} (5, 5);
        \draw[ ->, blue] [shift={(9, 0)}]  (5, 2)  to [bend right] node[midway, left] {$g_2$} (5, 1);
        \node at (9, 1) {(b)};
        \end{tikzpicture}}
        \caption{Two different presentations for the merge tree introduced in \Cref{example-m}. $R_2$ can be matched with either $G_1$ or $G_2$ since after $R_1$ they have already merged into one component.}
        \label{fig:3branch-tree-presentations}
    \end{figure}
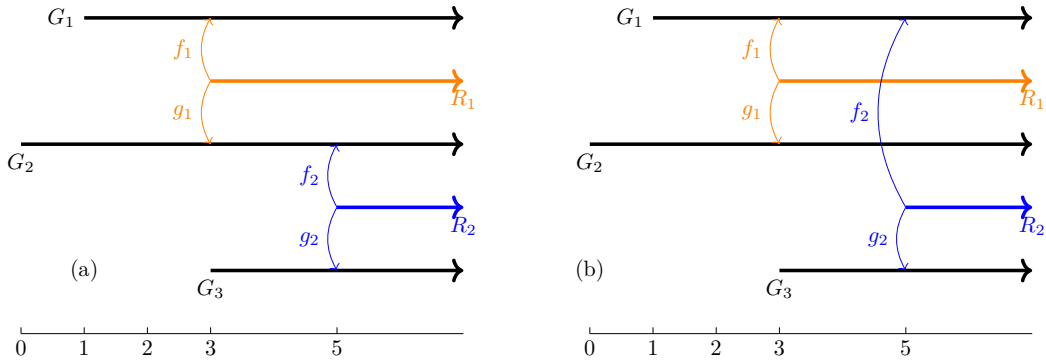
The corresponding presentation matrices for these are
\[
\begin{blockarray}{ccc}
 & 3 & 5 \\
\begin{block}{c(cc)}
  1 & 1 & 0\\
  0 & 1 & 1\\
  3 & 0 & 1\\
\end{block}
\end{blockarray}
\qquad \text{ and }\qquad
\begin{blockarray}{ccc}
 & 3 & 5 \\
\begin{block}{c(cc)}
  1 & 1 & 1\\
  0 & 1 & 0\\
  3 & 0 & 1\\
\end{block}
\end{blockarray} \enskip,
\]
but both of these presentations have the same label vector $L=[1,0,3;3,5]$.
\end{example}

We now present two key lemmas.

\begin{lemma}\label{lem:presentations-exist}
Every merge tree with $n$ leaf nodes has a presentation with $n$ generators and $n-1$ relations.
\end{lemma}

\begin{proof} 
We proceed by induction on the number of leaf nodes:
\Cref{presentation-example} shows that the claim holds for $n=2$. 
Suppose that the claim is true for some $k>2$ and  $M$ is a merge tree of $k+1$ leaf nodes. 
Let $t$ be the highest merge time of $M$, there exists $m$ groups of $k_1, k_2, \ldots, k_m > 0 $ leaf nodes whose merge times are not larger than $t$.
Notice that each $i$-th group of $k_i$ leaf nodes can be realized as a merge tree whose presentation consists of $k_i$ generators and $k_i - 1$ relations.
Let $G_i$ be a representative strand of the $i$-th group, ($ 1 \leq i \leq m$), by pairwise relating $G_1$ and $G_i$ at the time $t$ ($2\leq i \leq m$ ), one obtains a presentation that consists of $k + 1$ generators.
The number of relations is given by $(k_1 - 1 ) + \ldots + (k_m - 1) + (m - 1)$, which is $k$.
\end{proof}

We say that two presentations $P_M$ and $P_N$ are \define{compatible} if their underlying matrices $\overline{P}_M$ and $\overline{P}_N$ are the same.

\begin{lemma} \label{lemma-compatible-matrices-exist}
Any pair of merge trees $M$ and $N$ have compatible presentations.
\end{lemma}

\begin{proof}
    Given presentations of $M$ and $N$, we may add extra generators and relations to one of them, to obtain presentations  $P_M$ and $P_N$ for $M$ and $N$, respectively, with the same number of generators.  Write the matrices underlying $P_M$ and $P_N$ as $\overline{P}_M$ and $\overline{P}_N$, respectively, and denote their numbers of relations by $m$ and $n$.  Since $M$ and $N$ are merge trees, there exists $t\in \R$ such that $M_{t'}$ and $N_{t'}$ are singleton sets for all $t'\geq t$.  We construct compatible presentations for $\Tilde{P}_M$ and $\Tilde{P}_N$ with underlying matrix
$\begin{pmatrix} \bar P_M& \bar P_N\end{pmatrix}$:
For $\Tilde{P}_M$, we take the row labels and the first $m$ column labels to be the same as for $P_M$, with each of the last $n$ column labels equal to $t$; and for $\Tilde{P}_N$, we take the row labels and the last $n$ column labels be the same as for $P_N$, with each of first $m$ column labels equal to $t$.
Since at time $t$ all the strands of $M$ have been related to each other, $\Tilde{P}_M$ is indeed a presentation for $M$.  Similarly, $\Tilde{P}_M$ is a presentation for $N$.
\end{proof}

\begin{remark}\label{cor-noncompatible-forest}
More generally, two merge forests have compatible presentations if and only if they have the same number of connected components.
\end{remark}
\section{Presentation metric on merge trees}
We next introduce the $p$-presentation metrics merge trees, adapting the definnitions of \cite{LB2021}.
We first define a semi-metric on merge-trees (\Cref{def:label-and-semi-distance}) which measures the difference between merge trees in terms of the $\ell^p$-distance between the birth times of the generators and relations in a compatible presentation of $M$ and $N$.
Unfortunately, as \Cref{ex:CEX} shows, \Cref{def:label-and-semi-distance} fails to satisfy the triangle inequality, so we pass to sequences of merge trees in \Cref{def:p-path-distance} to get a genuine (pseudo)metric.

\begin{definition}[$p$-presentation semi-distance]\label{def:label-and-semi-distance}
If $P_M$ and $P_N$ are compatible presentations for merge trees $M$ and $N$, then for any $p\in[1,\infty]$ we define the \define{$p$-label distance} to be $\norm{L(P_M) - L(P_N)}_p$.
The \define{$p$-presentation semi-distance} between merge trees $M$ and $N$ is
\[
\hat{d}_{\PP}^p(M, N) = \inf\{ \norm{L(P_M) - L(P_N)}_p \mid P_M \text{ and } P_N
\text{ are compatible}\}.
\]
\end{definition}

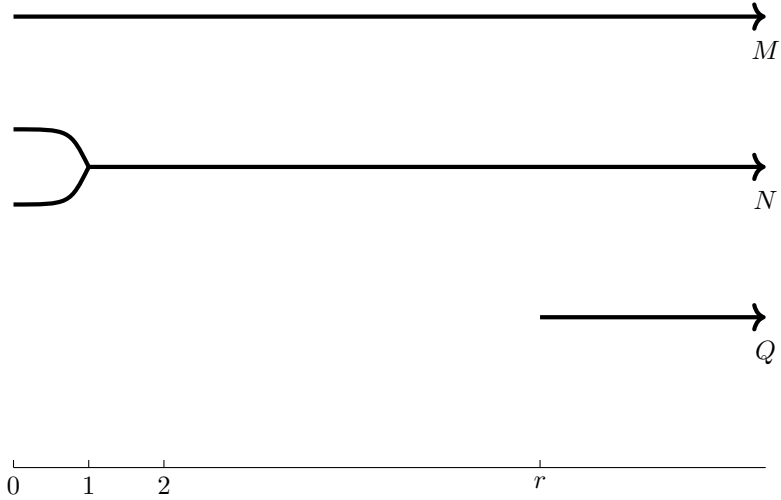
\begin{figure}
    \centering
    \begin{tikzpicture}
        \draw[ultra thick] (0, 3.5) .. controls (0.75, 3.5) .. (1, 3);
        \draw[ultra thick, ->] (0, 2.5) .. controls (0.75, 2.5) .. (1, 3) -- (10, 3)  node[below=5pt] {$N$};
        \draw[ultra thick, ->] (0, 5) -- (10, 5) node[below=5pt] {$M$};
        \draw[ultra thick, ->] (7, 1) -- (10, 1)  node[below=5pt] {$Q$};
        \draw (0, -1) -- (10, -1);
        \foreach \x in {0,1,2}{
        \draw (\x ,-1) node[anchor=north] {\x} -- (\x,-0.9);
        };
    \draw (7 ,-1) node[anchor=north] {$r$} -- (7,-0.9);
    \end{tikzpicture}
     \caption{Counterexample to the triangle inequality for $\hat{d}_{\PP}^p$.}
        \label{fig:CEX}
\end{figure}

\begin{example}\label{ex:CEX}
The following is a close analogue of \cite[Example 3.1]{LB2021}. In \Cref{fig:CEX} we have three merge trees $M,N$ and $Q$.
We claim that for $r$ large enough
\[
\sd^p(M, N) \leq 1, \qquad \sd^p(M, Q) = r, \quad \text{ and }\quad  \sd^p(N, Q) = \sqrt[p]{(r - 1)^p + 2r^p}.
\]
From this it follows that
\[
\sd^p(N, Q) \geq \sqrt[p]{(r - 1)^p + 2r^p} > 1 + r \geq \sd^p(N, M) + \sd^p(M, Q),
\]
and hence the triangle inequality does not hold for $\sd^p$.

Consider the following compatible presentations
\[
P_M :=
\begin{blockarray}{cc}
 & \epsilon \\
\begin{block}{c(c)}
  0 & 1 \\
  \epsilon & 1 \\
\end{block}
\end{blockarray}
\enskip,\qquad
P_N :=
\begin{blockarray}{cc}
 & 1 \\
\begin{block}{c(c)}
  0 & 1 \\
  0 & 1 \\
\end{block}
\end{blockarray}
\enskip,\qquad
P_Q :=
\begin{blockarray}{cc}
 & {r + \epsilon} \\
\begin{block}{c(c)}
  r & 1 \\
  {r + \epsilon} & 1 \\
\end{block}
\end{blockarray}.
\]

It is easy to see that $\sd^p(M, Q) = r$ by choosing compatible presentations with a single generator and no relations.
To see that $\sd^p(M, N) \leq 1$, observe that
\[
\sd^p(M, N) \leq \|L(P_M) - L(P_N)\|_p =
\sqrt[p]{\vert 1 - \epsilon \vert^p + \epsilon^p}.
\]
Setting $\epsilon=0$ for the presentation of $Q$ shows that
$\sd^p(N, Q) \leq \sqrt[p]{(r - 1)^p + 2r^p}$.
Any pair of compatible presentations $P'_N$ and $P'_Q$ for $N$ and $Q$ must contain a subpresentation of the form $P_Q$, since $P_N$ is minimal.
However, for all $\epsilon\geq 0$, the label vector difference $\norm{L(P'_M)-L(P'_N)}_p$ will be greater than $\sqrt[p]{r^p + (r +\epsilon)^p + (r+\epsilon-1)^p}$, which is minimized at $\epsilon=0$.
\end{example}

Although \Cref{ex:CEX} shows that $\sd^p$ is not a metric, we can repair this as follows:

\begin{definition}[$p$-presentation distance]\label{def:p-path-distance}
For merge trees $M$ and $N$ and $p\in[1,\infty]$ we define the \define{$p$-presentation distance} as
\[
d_{\PP}^p(M,N):=\inf \sum_{i=0}^{n-1} \sd^p(Q_i,Q_{i+1}),
\]
where we infimize over all finite sequences of merge trees $M=Q_0,\ldots, Q_n=N$.
\end{definition}

The following result is a close analogue of \cite[Proposition 3.6]{LB2021}.  It has essentially the same proof and will play a similarly important role in our arguments.
\begin{lemma}[Largest bounded metric]\label{lem:sd-universal-d}\mbox{}
\begin{itemize}
    \item[(i)] $d_{\PP}^p$ is a metric on isomorphism classes of merge trees.
    \item[(ii)] $d_{\PP}^p$ is the largest such metric bounded above by $\sd^p$.
\end{itemize}
\end{lemma}

\begin{proof}
By \Cref{lemma-compatible-matrices-exist}, the presentation distance between any two merge trees is finite. 
It then follows from \Cref{bl-lemma} in the appendix (reproduced from \cite[Rmk.~3.4]{LB2021} for convenience) that $d_{\PP}^p$ is the largest pseudometric bounded above by $\sd^p$. 
Finally, \cref{lem:iso-proof} shows that if $d_{\PP}^p(M,N)=0$, then $M$ is isomorphic to $N$, which finishes the proof.
\end{proof}

To formalize how varying $p$ increases the sensitivity of this distance, we recall the fundamental property of $\ell^p$-norms: for any vector $x$, $\| x \|_p \geq \| x \|_q$ whenever $p\leq q$.

\begin{proposition}\label{cor:comparing-p} For any pair of merge trees $M$ and $N$ and for all $1 \leq p \leq q \leq \infty$, we have $d^p_{\PP}(M, N) \geq d^q_{\PP}(M, N).$
\end{proposition}

\subsection{Equality of the $\infty$-presentation distance and interleaving distance}
We now prove \Cref{Thm:Infty_Pres_Dist_Equals_Interleaving_Distance}, which says that the $\infty$-presentation distance is equal to the interleaving distance.  First, we define the interleaving distance.

\begin{definition}\label{def:interleaving}
For $\epsilon >0$, there is a shift functor $S_\epsilon : \R \to \R$ given by $t \mapsto t + \epsilon$.
An $\epsilon$-\define{interleaving} between persistent sets $M$ and $N$ is given by a pair of natural transformations $\varphi : M \to N S_\epsilon$ and $\psi : N \to M S_\epsilon$ so that
the following diagrams commute for all $s \in \R$,
\[
\begin{tikzcd}
M(s) \ar[r] \ar[rd] & M(s + \epsilon) \ar[r] \ar[rd] & M(s + 2\epsilon) \\
N(s) \ar[r] \ar[ru] & N(s + \epsilon) \ar[r] \ar[ru] & N(s + 2\epsilon).
\end{tikzcd}
\]
The \define{interleaving distance} between two persistent sets $M$ and $N$, and hence two merge trees, is defined as
\[
d_I(M,N) := \inf \{\epsilon \mid M \text{ and } N \text{ are } \epsilon\text{-interleaved}\}.
\]
\end{definition}

\begin{proof}[Proof of \Cref{Thm:Infty_Pres_Dist_Equals_Interleaving_Distance}]
Let $M$ and $N$ be two merge trees. By \Cref{lem:sd-universal-d}(ii) it suffices to prove that $\sd^{\infty}=d_I$ because $d_I$ satisfies the triangle inequality and $d_{\PP}^{\infty}$ is the largest pseudometric bounded above by $\sd^{\infty}$.

Suppose there exists an
$\epsilon$-interleaving between $M$ and $N$ : $\varphi : M \to N  S_\epsilon$ and
$\psi : N \to M S_\epsilon$.  
Let $\overline P_M$ (resp. $\overline P_N$) be the underlying matrix for a presentation of $M$ (resp. $N$), whose generators and relations are $G_M$ and $R_M$ (resp. $G_N$ and $R_N$).
Here we slightly abuse notation by using generators and relations as row and column labels.
We define a matrix $P_Z$ as follows, 
\[
P_Z :=
\begin{blockarray}{ccccc}
 & R_M & R_N & G_M S_\epsilon & G_N S_\epsilon \\
\begin{block}{c(cccc)}
  G_M & \overline P_M & 0 & I & P_\psi  \\
  G_N & 0 & \overline P_N & P_\varphi & I \\
\end{block}
\end{blockarray}
\]
where $P_\psi$ is defined by
\[
(i, j) \mapsto
\begin{cases}
1 & \text{if } \psi(\pi G_i^N)(g_i^N) = g_j^M, \\
0 & \text{otherwise,}
\end{cases}
\]
and  $P_\varphi$ is defined by
\[
(i, j) \mapsto
\begin{cases}
1 & \text{if } \varphi(\pi G_i^M)(g_i^M) = g_j^N, \\
0 & \text{otherwise.}
\end{cases}
\]
Here $g_i^N$ represents the element of $N(\pi G_i^N)$ representing the
strand $G_i^N$. 
Moreover $G_M S_{\epsilon}$ denotes the collection of strands in $G_M$ that are obtained by $\epsilon$-shifting, likewise for $G_N S_{\epsilon}$.
By construction, $P_\psi$ and $P_\varphi$ have exactly one nonzero entry per column.
Using this matrix $P_Z$, one can obtain a pair of compatible presentations $P_M$ and $P_N$ for $M$ and
$N$ respectively by relabeling $P_Z$ as follows:
\[
P_M :=
\begin{blockarray}{ccccc}
 & R_M & R_N S_{\epsilon} & G_M S_{2\epsilon} & G_N S_{\epsilon} \\
\begin{block}{c(cccc)}
  G_M & \overline P_M & 0 & I & P_\psi  \\
  G_N S_{\epsilon} & 0 & \overline P_N & P_\varphi & I \\
\end{block}
\end{blockarray}\enskip,
\]
and
\[
P_N :=
\begin{blockarray}{ccccc}
 & R_M S_{\epsilon} & R_N & G_M S_{\epsilon} & G_N S_{2\epsilon} \\
\begin{block}{c(cccc)}
    G_M S_{\epsilon} & \overline P_M & 0 & I & P_\psi  \\
  G_N & 0 & \overline P_N & P_\varphi & I \\
\end{block}
\end{blockarray}\enskip.
\]
Since each of the generator and relation birth times for $P_M$ and $P_N$ differ exactly by $\epsilon$, we can conclude that  $\| L(P_M) - L(P_N) \|_\infty = \epsilon$.
Hence we have shown that $d_I(M, N) \geq \sd^\infty(M, N)$.

Now suppose that there exists a pair of compatible presentations $P_M$ and $P_N$ for $M$ and $N$
respectively, 
such that $\| L(P_M) - L(P_N) \|_\infty = \epsilon$.  
This hypothesis guarantees generators of $P_M$ are born within $\epsilon$ of generators of $P_N$, and likewise for their relations. This allows us to construct functorial mappings:
\[
\begin{tikzcd}
\bigsqcup R_i^M \ar[r, shift right] \ar[r, shift left] \ar[d, "\alpha"] &
\bigsqcup G_i^M \ar[r] \ar[d, "\beta"]&
M \ar[d, dashed, "\varphi"] \\
\bigsqcup R_i^N S_\epsilon \ar[r, shift right] \ar[r, shift left] &
\bigsqcup G_i^N S_\epsilon \ar[r] &
N S_\epsilon
\end{tikzcd}
\]
where $\alpha$ and $\beta$ are defined by taking $G_i^M \mapsto G_i^N S_{\epsilon}$ and $R_j^M \mapsto R_j^N  S_{\epsilon}$.
Here $\varphi : M \to N S_\epsilon$ is obtained via the universal property of the coequalizer.
Commutativity of the left square in the diagram follows from the fact that $P_M$ and $P_N$ have the same underlying matrix.
Similarly, we obtain a map $\psi : N \to M S_\epsilon$.
The uniqueness of the universal property guarantees $(\varphi, \psi)$ to be an interleaving pair.  
This shows that there exists a $\epsilon$-interleaving between $M$ and $N$, that is, 
$d_I(M, N) \leq \sd^\infty(M, N)$.
\end{proof}

\subsection{Comparison with Wasserstein distance on barcodes}

Now that metric properties of the $p$-presentation distance $d_{\PP}^p$ have been established, we turn our attention to the relation between this distance and the $p$-Wasserstein distance $d_{\WW}^p$ on barcodes.  Specifically, we prove \cref{thm:main-theorem-summary}\,(ii), which says that for $p\in [1, \infty]$ and any merge trees $M$ and $N$, we have \[d_{\WW}^p(\B(M),\B(N)) \leq d_{\PP}^p(M, N).\]

\begin{lemma}\label{lem:H_0-bound}
For merge trees $M$ and $N$,
\[d_{\PP}^p(M, N) \geq d_{\PP}^p\big( H_0(M), H_0(N) \big),\]
where $d_{\PP}^p\big( H_0(M), H_0(N) \big)$ denotes the $\ell^p$-distance on persistent modules, as defined in \cite{LB2021}.
\end{lemma}

\begin{proof}
Let $M$ and $N$ be merge trees.
We first show that $\sd^p(M,N)\geq \sd^p(H_0(M),H_0(N))$, where
$\sd^p\big( H_0(M), H_0(N) \big)$ denotes the $p$-presentation semi-distance on persistent modules, as defined in \cite{LB2021}.
If $P_M$, $P_N$ are compatible presentations for $M$ and $N$,
\[
\begin{tikzcd}
\bigsqcup_{i}R^M_{i}  \ar[r, yshift=.7ex, "f^M"] \ar[r, yshift=-.7ex, "g^M"'] &
\bigsqcup_{j} G^M_{j} \ar[r, dashed] & M, \end{tikzcd}
\quad \text{ and  } \quad
\begin{tikzcd}
\bigsqcup_{i}R^N_{i}  \ar[r, yshift=.7ex, "f^N"] \ar[r, yshift=-.7ex, "g^N"'] &
\bigsqcup_{j} G^N_{j} \ar[r, dashed] & N,
\end{tikzcd}
\]
applying $H_0$ will yield compatible presentations for $H_0(M)$ and $H_0(N)$,
\[
\begin{tikzcd}
    \bigoplus_{i} H_0(R^M_{i})  \ar[rr, "{f^M}^{*} - {g^M}^{*}"] & & \bigoplus_{j} H_0(G^M_{j}) \ar[r, dashed] & H_0(M), \\
    \bigoplus_{i} H_0(R^N_{i})  \ar[rr, "{f^N}^{*} - {g^N}^{*}"] & & \bigoplus_{j} H_0(G^N_{j}) \ar[r, dashed] & H_0(N),
\end{tikzcd}\]
whose $p$-distance is $\norm{L(P_M) - L(P_N)}_p$.
This implies that
$\sd^p(M, N) \geq \sd^p\big( H_0(M), H_0(N) \big)$.
This is sufficient to prove the statement because we know that
\sloppy $\sd^p(H_0(M),H_0(N))\geq d_{\PP}^p(H_0(M),H_0(N))$, by \cite[Prop.~3.3]{LB2021}.
Since $d_{\PP}^p(H_0(\bullet),H_0(\bullet))$ is a pseudometric on merge trees,
 \Cref{lem:sd-universal-d}(ii) implies that $d^p_{\PP}(M,N)\geq d_{\PP}^p(H_0(M),H_0(N))$.
\end{proof}

\begin{remark}
The inequality from \Cref{lem:H_0-bound} can be strict because non-isomorphic merge trees can have isomorphic persistent homology modules.
This was demonstrated in \cite{curry2018fiber}, where the following example was considered.
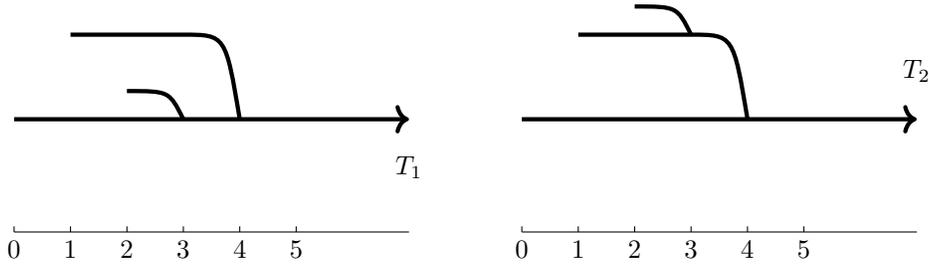
\[
    \begin{tikzpicture}[scale=0.75]
        \draw[->, ultra thick] (0, 0) -- (7, 0) node[anchor=north east, right, below=10pt] {$T_1$};
        \draw[ultra thick] (1, 1.5) -- (3, 1.5) .. controls (3.75, 1.5) .. (4, 0);
        \draw[ultra thick] (2, 0.5) .. controls (2.75, 0.5) .. (3, 0);
        \draw[->, ultra thick] [shift={(9, 0)}] (0, 0) -- (7, 0) node[anchor=north east, right, above=10pt] {$T_2$};
        \draw[ultra thick] [shift={(9, 0)}] (1, 1.5) -- (3, 1.5) .. controls (3.75, 1.5) .. (4, 0);
        \draw[ultra thick] [shift={(9, 1.5)}]  (2, 0.5) .. controls (2.75, 0.5) .. (3, 0);
        \draw[thin] (0, -2) -- (7, -2);
        \foreach \x in {0,1,2, 3, 4, 5}{
        \draw (\x ,-2) node[anchor=north] {\x} -- (\x,-1.9);
        };
        \draw[thin][shift={(9, 0)}]  (0, -2) -- (7, -2);
        \foreach \x in {0,1,2, 3, 4, 5}{
        \draw [shift={(9, 0)}]  (\x ,-2) node[anchor=north] {\x} -- (\x,-1.9);
        };
     \end{tikzpicture}
\]

\end{remark}

Continuing our comparison of the $p$-presentation distance on merge trees with metrics in persistent homology, we recall the definition of the Wasserstein distance on barcodes.

\begin{definition}\label{def:wasserstein}
A \define{barcode} $\mathcal B$ is a finite collection of intervals $\{I\}$ in $\R$.
A matching between barcodes $\mathcal B$ and $\mathcal C$ consists of a choice of subsets
 $\mathcal{B}' \subset \mathcal{B}$,and $\mathcal{C}' \subset \mathcal{C}$ and a bijection
 $\sigma: \mathcal{B'} \to \mathcal{C'}$.
For any $p \in [1,\infty]$ we define the \define{$p$-cost} of $\sigma$ as
\[
\cost{\sigma, p} = \left\{
\begin{aligned}
        \left( \sum\limits_{\substack{ I \in \mathcal B, \\ \sigma(I) = J }}{\norm{I - J}_p^p + }  \sum\limits_{I \in \Delta}{\norm{I - \Mid(I)}}_p^p \right)^{1/p}, & \text{ when } 1 \leq p < \infty \\
        \max\left\{ \max\limits_{\substack{ I \in \mathcal B, \\ \sigma(I) = J }}{\norm{I - J}_\infty, \max\limits_{I \in \Delta }{\norm{I - \Mid(I)}}_\infty } \right\}, & \text{ when } p = \infty,
\end{aligned}
\right \},
\]
where $\norm{I-J}_p$ is the $\ell^p$-norm between intervals $I = [a, b)$ and $J=[c,d)$ viewed as vectors $(a, b)$ and $(c,d)$ in $\R^2$, $\Mid(I) := \left[\tfrac{a+b}{2}, \tfrac{a+b}{2}\right)$ is the empty interval at the midpoint of $I$, and $\Delta$ denotes all the intervals unmatched by $\sigma$ in $\mathcal B \sqcup \mathcal C$.
The \define{Wasserstein $p$-distance} between barcodes $\mathcal B$ and $\mathcal C$ is then defined as the infimum of $p$-costs over all possible matchings, i.e.,
\[
d^p_{\WW}(\mathcal B, \mathcal C) = \inf\limits_{\sigma}{\cost{\sigma, p}}.
\]
The distance $d^\infty_\WW$ is called the \define{bottleneck distance}.
\end{definition}

\Cref{lem:H_0-bound} and \cite[Theorem 1.1]{LB2021} together establish \cref{thm:main-theorem-summary}\,(ii).

\section{Stability and universality}

In this section we consider two of the most important properties of the $p$-presentation metric on merge trees: stability and universality.
To motivate stability, we consider another matrix-based distance on merge trees known as the
$p$-cophenetic distance, which we show is \emph{not} stable for $p\in [1,\infty)$.

\paragraph*{Comparison with cophenetic distances}
\begin{definition}\label{def:cophenetic}
Given a merge tree $M$ together with a surjective ordered leaf node labelling $\pi= \{G_1, \ldots, G_k\}$, a \define{cophenetic vector} $C^{\pi}_M$ is an upper triangular matrix whose $(i, j)$-entry is the earliest merge 
time (height of the least common ancestor) of nodes $G_i$ and $G_j$.
For a pair of merge trees $M$ and $N$ together with a labelling $\pi$, the \define{labeled $p$-cophenetic distance} is defined as $\|C^\pi_M - C^\pi_N\|_p$, where we view these matrices as length $k(k+1)/2$-vectors.
One can then define the \define{$p$-cophenetic distance} between two merge trees as
\[
d^p_C(M, N) = \inf\limits_{\pi \in \Pi}{\|C_M^\pi - C_N^\pi\|_p},
\]
where $\Pi$ denotes the set of all surjective, ordered leaf labelings of $M$ and $N$.
\end{definition}

We now give a counterexample to stability of the cophenetic $p$-distance, when $p\neq \infty$, for cellular monotone functions; see \cpageref{def:cellular-monotone}, \Cref{def:cellular-monotone} for a reminder.

\begin{example} \label{ex:instability-cophenetic}
Let $X$ be the barycentric subdivision of the geometric $1$-simplex.
Consider the monotone cellular functions $f, g: X \to \R$ where $f \equiv 0$ and $g$ is $0$ on $0$-cells and $1$ on $1$-cells.
By inspection, the merge tree $M = \pi_0 S^{\uparrow}f$ has one leaf node and $N = \pi_0 S^{\uparrow}g$ has three leaf nodes, one for each $0$-cell, and a single internal node; see \Cref{fig:cophenetics}.

    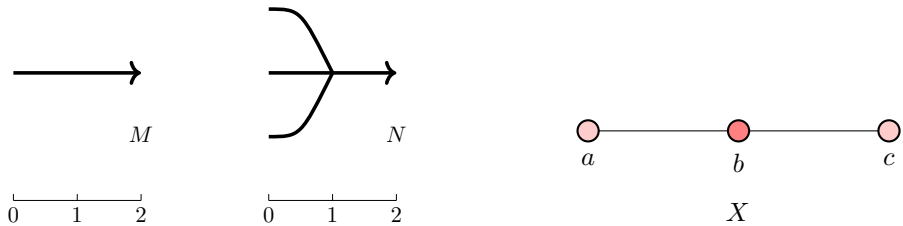
\begin{figure}[ht]
    \centering
\begin{minipage}[b]{0.4\textwidth}
    \centering
    \resizebox{\textwidth}{!}{%
        \begin{tikzpicture}
    \draw [->, ultra thick] (0, 4) -- (2, 4) node[anchor=north, right, below=20pt] {$M$} ;
    \foreach \x in {0,1,2} {
        \draw (\x ,2) node[anchor=north] {\x} -- (\x,2.1);
    };
    \draw (0, 2) -- (2, 2);
    \draw [->, ultra thick] [shift={(4, 0)}] (0, 4) -- (2, 4) node[anchor=north, right, below=20pt] {$N$};
    \draw [ultra thick] [shift={(4, 0)}] (0, 5) .. controls (0.5, 5) .. (1, 4);
    \draw [ultra thick] [shift={(4, 0)}] (0, 3) .. controls (0.5, 3) .. (1, 4);
    \foreach \x in {0,1,2} {
        \draw [shift={(4, 0)}]  (\x ,2) node[anchor=north] {\x} -- (\x,2.1);
    };
    \draw [shift={(4, 0)}] (0, 2) -- (2, 2);
    \end{tikzpicture}

    }
    \end{minipage}\qquad \qquad
    \begin{minipage}[b]{0.4\textwidth}
    \centering
    \begin{tikzpicture}
  \colorlet{circle edge}{black}
  \coordinate (A) at (-2,4);
  \coordinate (B) at (0, 4);
  \coordinate (C) at (2, 4);
  \coordinate (D) at ($(A)!.5!(B)$);
  \coordinate (E) at ($(B)!.5!(C)$);
  \draw (A) node[below=5pt] {$a$} -- (B) node[below=5pt] {$b$} -- (C) node[below=5pt] {$c$};
  \fill[filled] (A) circle (4pt);
  \fill[outline, fill=red!50] (B) circle (4pt);
  \fill[filled] (C) circle (4pt);
  \node[anchor=north, below =5pt] at (current bounding box.south) {$X$};
\end{tikzpicture}
\end{minipage}
\caption{Two merge trees are shown at left, associated to two functions on the cell complex $X$, at right. In \Cref{ex:instability-cophenetic} these give a counterexample to stability of the $p$-cophenetic distance.}
\label{fig:cophenetics}
\end{figure}

By allowing redundant labels for the one leaf node in $M$, we have cophenetic vectors
\[
C_M =
\begin{bmatrix}
0 & 0 & 0 \\
\ast & 0 & 0 \\
\ast & \ast & 0
\end{bmatrix}
\quad \text{ and }\qquad
C_N=
\begin{bmatrix}
0 & 1 & 1 \\
\ast & 0 & 1 \\
\ast & \ast & 0
\end{bmatrix}.
\]
From \cite{intrinsic_interleaving} we know that $d^\infty_C(M_f, M_g) = d_I(M_f, M_g)$, which is stable \cite{interleaving-distance-merge-trees}, i.e.,
\[
d^\infty_C(M_f, M_g)=d_I(M_f, M_g) \leq \|f -g\|_\infty.
\]
However, for $p\in [1, \infty)$, we have $d^p_C(M_f, N_g) = \sqrt[p]{3}$, which is larger than $\|f - g\|_p=\sqrt[p]{2}$.
\end{example}

\subsection{Stability}

\begin{definition}
A distance $d$ on merge trees is said to be \define{$p$-stable} if whenever
$f$ and $g$ are monotone cellular functions on a regular cell complex, the
associated merge trees $M= \pi_0 S^{\uparrow} f$ and
$N = \pi_0 S^{\uparrow} g$ satisfy $d(M,N)\leq \|f- g\|_p$.
\end{definition}

We show that $p$-presentation distance is $p$-stable.
\begin{proof}[Proof of \Cref{thm:main-theorem-summary}(i)]
We start by labeling vertices of $X$ by $\sigma_1, \ldots, \sigma_k$ and edges
by $\tau_1, \ldots, \tau_l$.
Consider the $(k\times l)$-matrix $P$ where
\[
P(i, j) :=
\begin{cases}
1 & \text{if } \sigma_i \subseteq \tau_j, \\
0 & \text{otherwise}.
\end{cases}
\]
We then define labeled matrices $P_M$, $P_N$ with underlying matrix $P$,
\[
P_M :=
\begin{blockarray}{ccccc}
 & \rho^M_1 & \cdots & \rho^M_l\\
\begin{block}{c(cccc)}
  \gamma^M_1 & & & \\
  \vdots & & & \\
  \gamma^M_k & & & \\
\end{block}
\end{blockarray}\enskip
\quad \text{and}
\quad
P_N :=
\begin{blockarray}{ccccc}
 & \rho^N_1 & \cdots & \rho^N_l \\
\begin{block}{c(cccc)}
    \gamma^N_1 & & & \\
  \vdots & & & \\
  \gamma^N_k & & & \\
\end{block}
\end{blockarray}\enskip,
\]
where $\gamma^M_i = f(\sigma_i)$ and $\rho^M_j = f(\tau_j)$, likewise for $P_N$.  
By definition of a regular cell complex each column of $P_M$ and $P_N$
must have exactly two ones.
Moreover, $P_M$ and $P_N$ are compatible presentation matrices of their
respective merge trees:
Monotonicity guarantees that we can obtain presentations for merge trees
$M$, $N$ with the generators and relations as described in $P_M$ and $P_N$.
By construction we have that
$\| L(P_M) - L(P_N) \|_p \leq \|f - g\|_p$ and by definition of
$\sd^p$ it follows that
$d_{\PP}^p(M, N) \leq \sd^p(M, N)  \leq \|f - g\|_p.$
\end{proof}

\begin{example}
If we present the merge trees $M$ and $N$ from \Cref{ex:instability-cophenetic}
using three generators and two relations, the corresponding presentation matrices are
\[ P_{M}:
\begin{blockarray}{ccc}
                & 0 & 0 \\
            \begin{block}{c(cc)}
            0 & 1 & 0 \\
            0 & 1 & 1 \\
            0 & 0 & 1 \\
            \end{block}
\end{blockarray} \qquad \text{ and } \qquad
P_{N}:
\begin{blockarray}{ccc}
                & 1 & 1 \\
            \begin{block}{c(cc)}
            0 & 1 & 0 \\
            0 & 1 & 1 \\
            0 & 0 & 1 \\
            \end{block}
\end{blockarray}
\]
and the label vectors are $L(P_M):=[0,0,0;0,0]$ and $L(P_M):=[0,0,0;1,1]$.
Thus the $p$-label distance is $\sqrt[p]{2}$, which equals the $\ell^p$-distance between $f$ and $g$. Consequently,
$\sd^p(M_f, N_g) \leq \sqrt[p]{2}  = \|f-g\|_p.$ Hence, $d_{\PP}^p(M_f, N_g) \leq \|f-g\|_p $, thus illustrating one advantage of this metric over the cophenetic $p$-distance.
\end{example}

\subsection{Universality}

In this section we prove that the $p$-presentation metric is universal among $p$-stable metrics. Here "universal" can be interpreted as maximal or final in a certain poset category of pseudometrics.
In order to prove \Cref{Thm:Universality}, we need the following lemma:

\begin{lemma}[Geometric lifting]
\label{lem-geometric-lifting}
If $M$ and $N$ are merge trees with compatible presentations $P_M$ and $P_N$ 
such that $\| L(P_M) - L(P_N) \|_p = \epsilon$, 
then there exists a regular cell complex $X$ and monotone cellular functions
$f : X \to \R$ and $g : X \to \R$ such that
\begin{enumerate}
    \item[(i)] $M \cong \pi_0 S^{\uparrow}f$, $N \cong \pi_0 S^{\uparrow}g$, and
    \item[(ii)] $\|f - g \|_p = \epsilon$.
\end{enumerate}
\end{lemma}
\begin{proof}
Let $P_M$ and $P_N$ be compatible presentations for $M$ and $N$ with presentation matrices 
\[
P_M :=
\begin{blockarray}{ccccc}
 & \rho_1^M & \cdots & \rho_l^M \\
\begin{block}{c(cccc)}
  \gamma_1^M & & & \\
  \vdots & & & \\
  \gamma_k^M & & & \\
\end{block}
\end{blockarray}\qquad \text{ and } \qquad
P_N :=
\begin{blockarray}{ccccc}
 & \rho_1^N & \cdots & \rho_l^N \\
\begin{block}{c(cccc)}
  \gamma_1^N & & & \\
  \vdots & & & \\
  \gamma_k^N & & & \\
\end{block}
\end{blockarray} \,.
\]
Using the underlying matrix for either presentation we construct a cell complex
$X$ as follows: For each row $i$, add a $0$-cell $\sigma_i$ to $X$, and for
each column $j$, attach a $1$-cell $\tau_j$ so that
$P_M(i, j) = 1\Rightarrow \sigma_i \subseteq \tau_j$.  
We then construct cellular functions $f, g : X \to \R$ using the birth times
for the generators and relations for $P_M$ and $P_N$, respectively.  
More precisely, $f(\sigma_i) = \gamma^M_i$ and $f(\tau_j) = \rho^M_j$, likewise for $g$.  
Since $\| L(P_M) - L(P_N)\|_p = \epsilon$ we have that $\|f - g \|_p = \epsilon$ as well.
\end{proof}

\begin{proof}[Proof of \Cref{Thm:Universality}]
It suffices to prove that for any $p$-stable distance $d$, we have $d\leq \sd^p$.  
Let $M, N$ be merge trees with $\sd^p(M, N) = \epsilon$ and let
$\epsilon' > \epsilon$ be arbitrary.
By \Cref{lem-geometric-lifting}, we can find
$f, g : X \to \R$
such that $M \cong \pi_0 S^{\uparrow} f$,
$N \cong \pi_0 S^{\uparrow} g$, and
$\|f - g\|_p = \epsilon'$.
By assumption $d(M, N) \leq \|f - g\|_p = \epsilon'$ and by letting
$\epsilon' \to \epsilon$, we have that $d(M, N) \leq \sd^p(M, N)$. 
By \Cref{lem:sd-universal-d}, we know that $d\leq d_{\PP}^p$.
\end{proof}


\bibliography{Presentation-main}

\begin{thebibliography}{10}

\bibitem{agarwal2018computing}
Pankaj~K Agarwal, Kyle Fox, Abhinandan Nath, Anastasios Sidiropoulos, and Yusu
  Wang.
\newblock Computing the {G}romov-{H}ausdorff distance for metric trees.
\newblock {\em ACM Transactions on Algorithms (TALG)}, 14(2):1--20, 2018.
\newblock \href {https://doi.org/10.1145/3185466} {\path{doi:10.1145/3185466}}.

\bibitem{bauer2019space}
Ulrich Bauer.
\newblock The space of {R}eeb graphs.
\newblock \textit{Workshop on Topology, Computation, and Data Analysis},
  Schloss Dagstuhl, 2019.

\bibitem{bauer2020universality}
Ulrich Bauer, Magnus~Bakke Botnan, and Benedikt Fluhr.
\newblock Universality of the bottleneck distance for extended persistence
  diagrams.
\newblock {\em arXiv preprint}, 2020.
\newblock \href {http://arxiv.org/abs/2007.01834} {\path{arXiv:2007.01834}}.

\bibitem{bauer2014measuring}
Ulrich Bauer, Xiaoyin Ge, and Yusu Wang.
\newblock Measuring distance between {R}eeb graphs.
\newblock In {\em Proceedings of the thirtieth annual symposium on
  Computational geometry}, pages 464--473, 2014.
\newblock \href {https://doi.org/10.1145/2582112.2582169}
  {\path{doi:10.1145/2582112.2582169}}.

\bibitem{bauer2020reeb}
Ulrich Bauer, Claudia Landi, and Facundo M{\'e}moli.
\newblock The {R}eeb graph edit distance is universal.
\newblock {\em Foundations of Computational Mathematics}, pages 1--24, 2020.
\newblock \href {https://doi.org/10.1007/s10208-020-09488-3}
  {\path{doi:10.1007/s10208-020-09488-3}}.

\bibitem{BL2020:diagrams}
Ulrich Bauer and Michael Lesnick.
\newblock Persistence diagrams as diagrams: A categorification of the stability
  theorem.
\newblock In {\em Topological Data Analysis}, pages 67--96. Springer
  International Publishing, 2020.
\newblock \href {https://doi.org/10.1007/978-3-030-43408-3_3}
  {\path{doi:10.1007/978-3-030-43408-3_3}}.

\bibitem{bauer2015strong}
Ulrich Bauer, Elizabeth Munch, and Yusu Wang.
\newblock Strong equivalence of the interleaving and functional distortion
  metrics for {R}eeb graphs.
\newblock In {\em 31st International Symposium on Computational Geometry (SoCG
  2015)}. Schloss Dagstuhl-Leibniz-Zentrum fuer Informatik, 2015.
\newblock \href {https://doi.org/10.4230/LIPIcs.SOCG.2015.461}
  {\path{doi:10.4230/LIPIcs.SOCG.2015.461}}.

\bibitem{bauer2021quasi}
Ulrich~J Bauer, H\aa vard~Bakke Bjerkevik, and Benedikt Fluhr.
\newblock Quasi-universality of {R}eeb graph distances.
\newblock {\em arXiv preprint}, 2021.
\newblock \href {http://arxiv.org/abs/1705.01690} {\path{arXiv:1705.01690}}.

\bibitem{LB2021}
H{\aa}vard~Bakke Bjerkevik and Michael Lesnick.
\newblock $\ell^p$-distances on multiparameter persistence modules.
\newblock {\em arXiv preprint}, 2021.
\newblock \href {http://arxiv.org/abs/2106.13589} {\path{arXiv:2106.13589}}.

\bibitem{blumberg2017universality}
Andrew~J Blumberg and Michael Lesnick.
\newblock Universality of the homotopy interleaving distance.
\newblock {\em arXiv preprint}, 2017.
\newblock \href {http://arxiv.org/abs/1705.01690} {\path{arXiv:1705.01690}}.

\bibitem{bollen2021reeb}
Brian Bollen, Erin~W. Chambers, Joshua~A. Levine, and Elizabeth Munch.
\newblock {R}eeb graph metrics from the ground up.
\newblock {\em arXiv preprint}, 2021.
\newblock \href {http://arxiv.org/abs/2110.05631} {\path{arXiv:2110.05631}}.

\bibitem{bubenik2015metrics}
Peter Bubenik, Vin de~Silva, and Jonathan Scott.
\newblock Metrics for generalized persistence modules.
\newblock {\em Foundations of Computational Mathematics}, 15(6):1501--1531,
  2015.
\newblock \href {https://doi.org/10.1007/s10208-014-9229-5}
  {\path{doi:10.1007/s10208-014-9229-5}}.

\bibitem{bubenik:algebraicwasserstein}
Peter Bubenik, Jonathan Scott, and Donald Stanley.
\newblock Exact weights, path metrics, and algebraic {W}asserstein distances.
\newblock {\em arXiv preprint}, 2018.
\newblock \href {http://arxiv.org/abs/1809.09654} {\path{arXiv:1809.09654}}.

\bibitem{bubenik2012:categorification}
Peter Bubenik and Jonathan~A Scott.
\newblock Categorification of persistent homology.
\newblock {\em Discrete \& Computational Geometry}, 51(3):600--627, 2014.
\newblock \href {https://doi.org/10.1007/s00454-014-9573-x}
  {\path{doi:10.1007/s00454-014-9573-x}}.

\bibitem{cardona2013cophenetic}
Gabriel Cardona, Arnau Mir, Francesc Rossell{\'o}, Lucia Rotger, and David
  S{\'a}nchez.
\newblock Cophenetic metrics for phylogenetic trees, after {S}okal and {R}ohlf.
\newblock {\em BMC bioinformatics}, 14(1):1--13, 2013.
\newblock \href {https://doi.org/10.1186/1471-2105-14-3}
  {\path{doi:10.1186/1471-2105-14-3}}.

\bibitem{JMLR-carlsson-memoli-10}
Gunnar Carlsson and Facundo M{{\'e}}moli.
\newblock Characterization, stability and convergence of hierarchical
  clustering methods.
\newblock {\em Journal of Machine Learning Research}, 11(47):1425--1470, 2010.
\newblock URL: \url{http://jmlr.org/papers/v11/carlsson10a.html}.

\bibitem{carr2003computing}
Hamish Carr, Jack Snoeyink, and Ulrike Axen.
\newblock Computing contour trees in all dimensions.
\newblock {\em Computational Geometry}, 24(2):75--94, 2003.
\newblock \href {https://doi.org/10.1016/S0925-7721(02)00093-7}
  {\path{doi:10.1016/S0925-7721(02)00093-7}}.

\bibitem{carr2004simplifying}
Hamish Carr, Jack Snoeyink, and Michiel van~de Panne.
\newblock Simplifying flexible isosurfaces using local geometric measures.
\newblock In {\em IEEE Visualization 2004}, pages 497--504, 2004.
\newblock \href {https://doi.org/10.1109/VISUAL.2004.96}
  {\path{doi:10.1109/VISUAL.2004.96}}.

\bibitem{chambers2021:truncated}
Erin~Wolf Chambers, Elizabeth Munch, and Tim Ophelders.
\newblock A family of metrics from the truncated smoothing of {R}eeb graphs.
\newblock In {\em 37th International Symposium on Computational Geometry},
  2021.
\newblock \href {https://doi.org/10.4230/LIPIcs.SoCG.2021.22}
  {\path{doi:10.4230/LIPIcs.SoCG.2021.22}}.

\bibitem{chazal2009proximity}
Fr{\'e}d{\'e}ric Chazal, David Cohen-Steiner, Marc Glisse, Leonidas~J Guibas,
  and Steve~Y Oudot.
\newblock Proximity of persistence modules and their diagrams.
\newblock In {\em Proceedings of the twenty-fifth annual symposium on
  Computational geometry}, pages 237--246, 2009.
\newblock \href {https://doi.org/10.1145/1542362.1542407}
  {\path{doi:10.1145/1542362.1542407}}.

\bibitem{chen2019generalized}
Yen-Chi Chen.
\newblock Generalized cluster trees and singular measures.
\newblock {\em The Annals of Statistics}, 47(4):2174--2203, 2019.
\newblock \href {https://doi.org/10.1214/18-AOS1744}
  {\path{doi:10.1214/18-AOS1744}}.

\bibitem{cohen2010lipschitz}
David Cohen-Steiner, Herbert Edelsbrunner, John Harer, and Yuriy Mileyko.
\newblock Lipschitz functions have ${L}^p$-stable persistence.
\newblock {\em Foundations of computational mathematics}, 10(2):127--139, 2010.
\newblock \href {https://doi.org/10.1007/s10208-010-9060-6}
  {\path{doi:10.1007/s10208-010-9060-6}}.

\bibitem{crawley2015decomposition}
William Crawley-Boevey.
\newblock Decomposition of pointwise finite-dimensional persistence modules.
\newblock {\em Journal of Algebra and its Applications}, 14(05):1550066, 2015.
\newblock \href {https://doi.org/10.1142/S0219498815500668}
  {\path{doi:10.1142/S0219498815500668}}.

\bibitem{curry2018fiber}
Justin Curry.
\newblock The fiber of the persistence map for functions on the interval.
\newblock {\em Journal of Applied and Computational Topology}, 2(3):301--321,
  2018.
\newblock \href {https://doi.org/10.1007/s41468-019-00024-z}
  {\path{doi:10.1007/s41468-019-00024-z}}.

\bibitem{curry2021:DMT}
Justin Curry, Haibin Hang, Washington Mio, Tom Needham, and Osman~Berat Okutan.
\newblock Decorated merge trees for persistent topology.
\newblock {\em To Appear in the Journal of Applied and Computational Topology},
  2022.
\newblock \href {http://arxiv.org/abs/2103.15804} {\path{arXiv:2103.15804}}.

\bibitem{curry2014:thesis}
Justin~Michael Curry.
\newblock {\em Sheaves, {C}osheaves and {A}pplications}.
\newblock PhD thesis, University of Pennsylvania, 2014.

\bibitem{categorified_reeb}
Vin de~Silva, Elizabeth Munch, and Amit Patel.
\newblock Categorified {R}eeb graphs.
\newblock {\em Discrete \& Computational Geometry}, 55(4):854--906, 2016.
\newblock \href {https://doi.org/10.1007/s00454-016-9763-9}
  {\path{doi:10.1007/s00454-016-9763-9}}.

\bibitem{deSilva2017:catflow}
Vin de~Silva, Elizabeth Munch, and Anastasios Stefanou.
\newblock Theory of interleavings on categories with a flow.
\newblock {\em Theory and Applications of Categories}, 33(21):583--607, 2018.
\newblock URL: \url{http://www.tac.mta.ca/tac/volumes/33/21/33-21.pdf}.

\bibitem{di2016edit}
Barbara Di~Fabio and Claudia Landi.
\newblock The edit distance for {R}eeb graphs of surfaces.
\newblock {\em Discrete \& Computational Geometry}, 55(2):423--461, 2016.
\newblock \href {https://doi.org/10.1007/s00454-016-9758-6}
  {\path{doi:10.1007/s00454-016-9758-6}}.

\bibitem{le2002random}
Thomas Duquesne and Jean-Fran\c{c}ois Le~Gall.
\newblock {\em Random trees, {L\'evy} processes and spatial branching
  processes}.
\newblock Number 281 in Ast\'erisque. Soci\'et\'e math\'ematique de France,
  2002.
\newblock URL: \url{http://www.numdam.org/item/AST_2002__281__R1_0/}.

\bibitem{farahbakhsh2021fr}
Elena Farahbakhsh~Touli.
\newblock Fr\'echet-like distances between two rooted trees.
\newblock {\em Journal of Algorithms and Computation}, 53(1):1--12, 2021.
\newblock \href {https://doi.org/10.22059/JAC.2021.81145}
  {\path{doi:10.22059/JAC.2021.81145}}.

\bibitem{farahbakhsh2019fpt}
Elena Farahbakhsh~Touli and Yusu Wang.
\newblock {FPT}-algorithms for computing {G}romov-{H}ausdorff and interleaving
  distances between trees.
\newblock In {\em European Symposium on Algorithms}, 2019.
\newblock \href {https://doi.org/10.4230/LIPIcs.ESA.2019.83}
  {\path{doi:10.4230/LIPIcs.ESA.2019.83}}.

\bibitem{flamm2002barrier}
Christoph Flamm, Ivo~L. Hofacker, Peter~F. Stadler, and Michael~T. Wolfinger.
\newblock Barrier trees of degenerate landscapes.
\newblock {\em Zeitschrift f{\"u}r Physikalische Chemie}, 2002.
\newblock \href {https://doi.org/10.1524/zpch.2002.216.2.155}
  {\path{doi:10.1524/zpch.2002.216.2.155}}.

\bibitem{intrinsic_interleaving}
Ellen {Gasparovic}, Elizabeth {Munch}, Steve {Oudot}, Katharine {Turner}, Bei
  {Wang}, and Yusu {Wang}.
\newblock Intrinsic interleaving distance for merge trees.
\newblock {\em arXiv preprint}, July 2019.
\newblock \href {http://arxiv.org/abs/1908.00063} {\path{arXiv:1908.00063}}.

\bibitem{GC1987}
Christopher {Gold} and Sean {Cormack}.
\newblock Spatially ordered networks and topographic reconstructions.
\newblock {\em International Journal of Geographical Information Systems},
  1(2):137--148, 1987.
\newblock \href {https://doi.org/10.1080/02693798708927800}
  {\path{doi:10.1080/02693798708927800}}.

\bibitem{hartigan81}
John~A. {Hartigan }.
\newblock Consistency of single linkage for high-density clusters.
\newblock {\em Journal of the American Statistical Association},
  76(374):388--394, 1981.
\newblock \href {https://doi.org/10.1080/01621459.1981.10477658}
  {\path{doi:10.1080/01621459.1981.10477658}}.

\bibitem{kashiwara2018persistent}
Masaki Kashiwara and Pierre Schapira.
\newblock Persistent homology and microlocal sheaf theory.
\newblock {\em Journal of Applied and Computational Topology}, 2(1):83--113,
  2018.
\newblock \href {https://doi.org/10.1007/s41468-018-0019-z}
  {\path{doi:10.1007/s41468-018-0019-z}}.

\bibitem{KWEON1994171}
In~So {Kweon} and Takeo {Kanade }.
\newblock Extracting topographic terrain features from elevation maps.
\newblock {\em CVGIP: Image Understanding}, 59(2):171--182, 1994.
\newblock \href {https://doi.org/10.1006/ciun.1994.1011}
  {\path{doi:10.1006/ciun.1994.1011}}.

\bibitem{lesnick2015theory}
Michael Lesnick.
\newblock The theory of the interleaving distance on multidimensional
  persistence modules.
\newblock {\em Foundations of Computational Mathematics}, 15(3):613--650, 2015.
\newblock \href {https://doi.org/10.1007/s10208-015-9255-y}
  {\path{doi:10.1007/s10208-015-9255-y}}.

\bibitem{interleaving-distance-merge-trees}
Dmitriy Morozov, Kenes Beketayev, and Gunther Weber.
\newblock Interleaving distance between merge trees.
\newblock {\em Proceedings of Topology-Based Methods in Visualization}, 2013.

\bibitem{Munch2019:cophenetic}
Elizabeth Munch and Anastasios Stefanou.
\newblock The {$l^\infty$}-cophenetic metric for phylogenetic trees as an
  interleaving distance.
\newblock In {\em Research in Data Science}, pages 109--127. Springer
  International Publishing, 2019.
\newblock \href {https://doi.org/10.1007/978-3-030-11566-1_5}
  {\path{doi:10.1007/978-3-030-11566-1_5}}.

\bibitem{patel2018generalized}
Amit Patel.
\newblock Generalized persistence diagrams.
\newblock {\em Journal of Applied and Computational Topology}, 1(3):397--419,
  2018.
\newblock \href {https://doi.org/10.1007/s41468-018-0012-6}
  {\path{doi:10.1007/s41468-018-0012-6}}.

\bibitem{perez2020c0persistent-a}
Daniel Perez.
\newblock On ${C}^0$-persistent homology and trees.
\newblock {\em arXiv preprint}, 2020.
\newblock \href {http://arxiv.org/abs/2012.02634} {\path{arXiv:2012.02634}}.

\bibitem{perez2020persistent-b}
Daniel Perez.
\newblock On the persistent homology of almost surely ${C}^0$ stochastic
  processes.
\newblock {\em arXiv preprint}, 2020.
\newblock \href {http://arxiv.org/abs/2012.09459} {\path{arXiv:2012.09459}}.

\bibitem{pont2021wasserstein}
Mathieu Pont, Jules Vidal, Julie Delon, and Julien Tierny.
\newblock Wasserstein distances, geodesics and barycenters of merge trees.
\newblock {\em IEEE Transactions on Visualization and Computer Graphics},
  28(1):291--301, 2021.
\newblock \href {https://doi.org/10.1109/TVCG.2021.3114839}
  {\path{doi:10.1109/TVCG.2021.3114839}}.

\bibitem{turner2017wasserstein}
Andrew {Robinson} and Katharine {Turner}.
\newblock Hypothesis testing for topological data analysis.
\newblock {\em Journal of Applied and Computational Topology}, 1(2):241--261,
  Dec 2017.
\newblock \href {https://doi.org/10.1007/s41468-017-0008-7}
  {\path{doi:10.1007/s41468-017-0008-7}}.

\bibitem{scoccola2020:thesis}
Luis~N Scoccola.
\newblock {\em Locally persistent categories and metric properties of
  interleaving distances}.
\newblock PhD thesis, The University of Western Ontario, 2020.

\bibitem{skraba2021wasserstein}
Primoz Skraba and Katharine Turner.
\newblock Wasserstein stability for persistence diagrams.
\newblock {\em arXiv preprint}, 2021.
\newblock \href {http://arxiv.org/abs/2006.16824} {\path{arXiv:2006.16824}}.

\bibitem{edit-distance-MTs-2020}
Raghavendra Sridharamurthy, Talha~Bin Masood, Adhitya Kamakshidasan, and Vijay
  Natarajan.
\newblock Edit distance between merge trees.
\newblock {\em IEEE Transactions on Visualization and Computer Graphics},
  26(3):1518--1531, 2020.
\newblock \href {https://doi.org/10.1109/TVCG.2018.2873612}
  {\path{doi:10.1109/TVCG.2018.2873612}}.

\bibitem{sridharamurthy2021comparative}
Raghavendra Sridharamurthy and Vijay Natarajan.
\newblock Comparative analysis of merge trees using local tree edit distance.
\newblock {\em IEEE Transactions on Visualization and Computer Graphics}, 2021.

\bibitem{turner2020medians}
Katharine Turner.
\newblock Medians of populations of persistence diagrams.
\newblock {\em Homology, Homotopy and Applications}, 22(1):255--282, 2020.
\newblock \href {https://doi.org/10.4310/HHA.2020.v22.n1.a15}
  {\path{doi:10.4310/HHA.2020.v22.n1.a15}}.

\end{thebibliography}

\appendix

\section{Supplement to proof of \Cref{lem:sd-universal-d}}
\begin{lemma}[{\cite[Remark 3.4]{LB2021}}] \label{bl-lemma}
Let $\hat d : X \times X \to [0, \infty)$ be a symmetric and reflexive function.
Define $d : X \times X \to [0, \infty)$ by 
\[d(x, y) := \inf \sum_{i = 0}^{n - 1} \hat d(z_i, z_{i + 1}),\]
where the infimum is taken over all finite sequences $z_0, \ldots, z_n$ in $X$ 
with $z_0 = x$ and $z_n = y$.
Then $d$ is the largest pseudometric bounded above by $\hat d$.
\end{lemma}

\noindent 
We reproduce the proof of 
\cite[Proposition 3.3]{LB2021}.

\begin{proof}
For $x \in X$, it follows that $d(x, x) \leq \hat d(x, x) = 0$.
For $x, y \in X$, observe that
\[d(x, y) = \inf \sum_{i = 0}^{n - 1} \hat d(z_i, z_{i + 1})
= \inf \sum_{i = 0}^{n - 1} \hat d(z_{i + 1}, z_{i}) = d(y, x)\]
since $\hat d$ is symmetric by hypothesis.
Let $x, y, w \in X$, and suppose $\delta > d(x, w)$ and $\delta' > d(w, y)$.
By definition there exist $x = z_0, \ldots, z_n = w$ and 
$w = z_0', \ldots, z_m' = y$ such that 
\[\sum_{i = 0}^{n - 1} \hat d(z_i, z_{i + 1}) < \delta
\quad \text{ and } \quad 
\sum_{i = 0}^{m - 1} \hat d(z_i', z_{i + 1}') < \delta'. \]
By definition of $d$, we see that 
\[
d(x, y) \leq 
\inf \sum_{i = 0}^{n - 1} \hat d(z_i, z_{i + 1}) + 
\inf \sum_{i = 0}^{m - 1} \hat d(z_i', z_{i + 1}') < \delta + \delta'.
\]
By arbitrariness of $\delta, \delta'$, conclude that 
$d(x, y) \leq d(x, w) + d(w, y)$.

Let $d'$ be a pseudometric bounded above by $\hat d$. Let $x, y \in X$ and 
suppose $\delta > d(x, y)$. By definition, there exist 
$x = z_0, \ldots, z_n = y$ such that 
\[
\sum_{i = 0}^{n - 1} \hat d(z_i, z_{i + 1}) < \delta.
\]
Now observe that 
\[
d'(x, y) \leq 
\sum_{i = 0}^{n - 1} d'(z_i, z_{i + 1}) \leq 
\sum_{i = 0}^{n - 1} \hat d(z_i, z_{i + 1}) <
\delta.
\]
By arbitrariness of $\delta > d(x, y)$, conclude that $d'(x, y) \leq d(x, y)$.
\end{proof}

\begin{lemma}[cf. {\cite[Proposition 3.6]{LB2021}}]\label{lem:iso-proof}
If $M$ and $N$ are merge trees with $d_{\PP}^p(M,N)=0$, then $M$ and $N$ are isomorphic.
\end{lemma}
\begin{proof}
By \cref{cor:comparing-p}, $d_{\PP}^\infty(M,N) \leq d_{\PP}^p(M,N)$, so it suffices to prove the result in the case that $p=\infty$.  By \cref{Thm:Infty_Pres_Dist_Equals_Interleaving_Distance}, we have $d_{\PP}^\infty(M,N)=d_I(M,N)$.  It is shown in \cite[Corollary 6.2]{lesnick2015theory} that if a pair of finitely presented multiparameter persistence modules have interleaving distance 0, then they are isomorphic.  In fact, essentially the same argument shows that if $d_I(M,N)=0$, then $M$ and $N$ are isomorphic.  
\end{proof}

\end{document}